\documentclass[journal, draftclsnofoot, onecolumn, 12pt]{IEEEtran}
\usepackage{graphics}
\usepackage{graphicx}
\graphicspath{{images/}}
\usepackage{tikz}
\usepackage{pgfplots}
\pgfplotsset{compat=1.13}
\usepackage{subcaption}
\usepackage{floatrow}
\newfloatcommand{capbtabbox}{table}[][\FBwidth]
\usepackage{multirow}
\usepackage{array}
\usepackage{booktabs}
\usepackage[nospace,noadjust]{cite}
\usepackage[nolist,printonlyused]{acronym}
\usepackage{comment}
\usepackage{amssymb,amsmath,amsfonts,amsthm,bm,mathtools,cuted}
\usepackage[ruled,vlined,linesnumbered]{algorithm2e}
\usepackage{braket}
\newcommand{\E}{\mathbb{E}}
\newcommand{\mc}[1]{\mathcal{#1}}
\newcommand{\mb}[1]{\mathbf{#1}}

\DeclareMathOperator*{\argmax}{arg\,max}

\newcommand{\tc}{\,|\,}
\newcommand{\AAB}[1]{\textcolor{blue}{{#1}}}

\newcommand{\MM}[1]{\textcolor{blue}{{#1}}}
\newcommand{\FS}[1]{\textcolor{blue}{{#1}}}

\theoremstyle{plain}

\newtheorem{theorem}{Theorem}

\begin{document}
\title{Power and subcarrier allocation in multicarrier NOMA-FD systems}
\author{\IEEEauthorblockN{Andrea Abrardo\IEEEauthorrefmark{1}, Marco Moretti\IEEEauthorrefmark{2}, Fabio Saggese\IEEEauthorrefmark{2}} \\
    \IEEEauthorblockA{\IEEEauthorrefmark{1}Dipartimento di Ingegneria dell'Informazione, University of Siena, Italy.} \\
    \IEEEauthorblockA{\IEEEauthorrefmark{2}Dipartimento di Ingegneria dell'Informazione, University of Pisa, Italy.}}
\maketitle

\begin{abstract}
In this paper, we study the problem of power and channel allocation for multicarrier non-orthogonal multiple access (NOMA) full duplex (FD) systems. In such a system there are multiple interfering users transmitting over the same channel and the allocation task is a non-convex and extremely challenging problem. The objective of our work is to propose a solution that is close to the theoretic optimum but is of limited complexity. Following a block coordinate descent approach, we propose two algorithms  based on the decomposition of the original allocation problem in  lower-complexity sub-problems, which can be solved in the Lagrangian dual domain with a great reduction of the computational load. Numerical results show the effectiveness of  approach we propose, which outperforms other schemes designed to address NOMA-FD allocation  and attains performance  similar to the optimal solution with much lower complexity.
\end{abstract}

% ================================================
\section{Introduction} \label{Sec:Intro}
% ================================================

With the advent of 5G, the coming years will see an explosive growth of mobile data traffic and a  dramatic increase in the number of mobile devices, calling for the introduction of revolutionary wireless technologies to sustain the ever-increasing demand for bandwidth and services. Full duplex (FD) and non-orthogonal multiple access (NOMA) are among the most promising key enabling technologies for 5G cellular systems~\cite{Wong2017}. 

The FD paradigm allows uplink and downlink transmissions to occur simultaneously on the same frequency channel and has the potential of doubling the spectral efficiency of conventional half-duplex communication systems~\cite{Sabharwal14}, provided that a sufficiently large part of the self-interference (SI) is canceled. Self-interference cancellation is mostly done in the analog domain, employing specific hardware~\cite{Zhang2015} and, for the complexity of the operation, it is customary to implement  FD at the base station only.

NOMA, whose basic idea was introduced in~\cite{Saito13}, multiplexes several users on the same spectral resources by exploiting channel diversity and can greatly increase device connectivity in comparison to traditional orthogonal multiple access (OMA) schemes. Since allocating more than one user on the same channel leads to severe co-channel interference, NOMA receivers perform successive interference cancellation (SIC) to remove part of the interference. 
SIC  is based on iterative decoding and cancelation of  interference and this leads to different problems when considering the uplink and the downlink segments, respectively~\cite{Ding2016a}. In fact, in the uplink direction, all signals are received by the base station, which is able to decode and cancel the various data streams according a given order with no particular problems. In the downlink direction, each user receives its signal plus all the signals intended for the other users allocated on the same channel, so that each user has to deal with this (potentially catastrophic) interference.  Accordingly, when it is possible, the receiver aims to first cancel the interference plaguing its reception and then decode its signal. This leads to a set of constraints on the minimum rate seen by the receiver for each data stream to be canceled, which greatly complicates the receiver task. A deeper investigation of these constraints is made in the remainder of the paper.

\MM{The advantages of combining together NOMA and FD  with respect to conventional systems can be summarized in enhanced flexibility,  fairer access  and larger spectral efficiency \cite{Ding2018,Zhang2019,Nguyen2019}. }
\MM{\subsection{Related work and motivation}}
Because of their nature, FD and NOMA  have in common that, since they are severely affected by interference,  their performance heavily relies on channel and power optimization. In fact, efficient resource allocation allows to fully capture the multi-user diversity of the system and delivers great gains with respect to more mature and conventional technologies.

\MM{Indeed, the problem of power and subcarrier allocation has been intensely researched in the last two decades. It originally dates back to the early 2000s, with the first relevant paper, the seminal work \cite{Wong1999}, being  published in 1999. Nevertheless, the field of optimization for wireless communications is greatly evolved in all these years mirroring the growth experienced by wireless communications in general. %In [1] the problem studied was how to allocate a pool of orthogonal resources among  different users. Published in 2006, 
A major breakthrough is represented by \cite{Yu2006}, which first showed that the non-convex channel allocation problem could be solved close to optimality in the Lagrangian dual domain. Since then, many other works have addressed the allocation problem in different scenarios and for different radio techniques, covering cross-layer solutions \cite{Su2008}, cognitive radio \cite{Xie2012}, small cell and heterogeneous networks \cite{Zhang2014}, cloud radio acces networks \cite{Peng2014} as well as MIMO systems \cite{Ng2012}, just to make few examples of very referenced works. Nevertheless, recent 5G-inspired systems, where radio resources are shared non-orthogonally among users, have opened entirely new fields to investigate, as it is witnessed by the large number of recent papers on resource allocation on D2D \cite{Liang2017,Wang2018,DellaPenda2019}, FD \cite{Zeng2018,Abrardo2018,Zhang2019a} and NOMA \cite{Zeng2019,Muhammed2019,Al-Eryani2019,Saggese2019}.}
\MM{Accordingly, the new set of problems requires the introduction of novel techniques and algorithms. In particular,  in this case the main challenge  is to constructively manage the inherent system multi-user interference arising from non-orthogonal access to the channel. Unfortunately, interference  makes the allocation problem not convex so that its solution requires the use of advanced and complex algorithms.   In \cite{Ding2018a} the authors have shown that the NOMA-FD approach is theoretically feasible and that can  yield large gains over half-duplex NOMA and orthogonal multiple access,  under the condition that the cochannel multi-user interference is properly managed.}

\MM{The resource allocation problem, central to the implementation of NOMA-FD systems, has been discussed in \cite{Sun2018,Sun2017,Nguyen2019a}. While \cite{Sun2018} considers multi-carrier NOMA, where co-channel users are assigned different codes, \cite{Sun2017} and \cite{Nguyen2019a} investigate power-domain NOMA, where users are only separated by their different levels of power at the receiver. Both these last two works first propose an optimal  benchmark, which has unpractical complexity because of  the nature of the problem,  and then study suboptimal  solutions with a lower computational complexity. For instance, \cite{Sun2017} finds the optimal channel and power allocation  by employing the framework of  \emph{monotonic optimization} and then proposes a simpler algorithm based on the continuous relaxation of the integral allocation variable, which is forced to assume integer values by  employing the \emph{barrier method}.  Although there are no theoretic limits to the number of users in a NOMA system,  this paper considers  a system where  two uplink and two downlink users are allocated on each frequency channel, as it is common for NOMA-FD single-input single-output (SISO) systems.  The  work in \cite{Nguyen2019a}, which studies a system where the base station is equipped with multiple antennas,  deals with the integer allocation variable in the same fashion: first it  relaxes the integral constraint,  so that the allocation variable can assume any value between zero and one and then the  problem is formulated in such a way that the relaxed variable is forced to take an integer value by introducing a large penalty for any  value different from zero and one. The performance of these heuristic methods tends to  depend  on the value of the penalty, which, if not properly chosen for each different scenario, may lead to results very far from the optimal solution. Moreover, the complexity of the proposed algorithms although much smaller than that the optimal solution is still prohibitively high.}

\MM{From a practical point of view the performance of NOMA-FD systems depends also on another aspect: availability of reliable channel state information (CSI). This topic is still in its infancy: in \cite{Aswathi2019} the authors studied the performance of a NOMA-FD system under the hypothesis of CSI errors and imperfect SIC and  \cite{Wei2017,Celik2019} deal with resource allocation with imperfect CSI in NOMA systems.}
 %  andThe second one reformulates the allocation problem by relaxing the integer condition on the channel allocation variable and, although suboptimal, manages to find a solution very close to the optimum. Nevertheless, also this second algorithm has a very large complexity that makes it not feasible in most practical load conditions. 
%\MM{ at In \cite{Aswathi2019}.  and the algorithm we study is original and our major contribution is that we managae to obtain near-optimal results at a fraction of the complexity of the theoretic optimum.}
%
%Recently, the seminal paper~ has proposed a NOMA-FD system that combines these two promising technologies by letting several downlink and uplink users share the same channel. The advantages of combining NOMA and FD are great, but they can be harvested only by a system capable to control the inherent interference by exploiting the system multi-user diversity. 
%Furthermore, both algorithms are formulated without taking into account the possibility that it might be more beneficial on certain channels to adopt an OMA approach rather than always going for NOMA. 
%
\subsection{Paper Contributions}  
\MM{The focus of our work is to solve the problem of channel and resource allocation for a SISO FD-NOMA system with the objective of maximizing the weighted users' sum rate, subject to power constraints. This non-convex problem is particularly difficult to solve and the few practical algorithms presented in the literature have still very large complexity.  With respect to the existing literature, we follow a different path to address the problem of dealing with binary integer allocation variables.  In fact, monotonic optimization and integer relaxation are not the only techniques that can be employed to address such complex problems. Recent literature~\cite{Yu2006},~\cite{Xiao17} has investigated the advantages of solving non-convex allocation problems in the dual domain and sometimes a successful approach at dealing with large complex problems is to try to reduce them in a sequence of more tractable subproblems~\cite{moretti2015}.}

\MM{In detail, the main contribution of this paper is a globally convergent algorithm able to solve the NOMA-FD allocation problem to a (good) local optimum. Following a block coordinate descent approach~\cite{Xu2017}, the optimization problem is split in two steps that are solved sequentially and  iterated until convergence:
\begin{enumerate}
\item one uplink user and one downlink user are assigned in an FD fashion to each radio channel; 
\item a second user to each uplink and downlink channel are assigned in order to maximizing the global utility, enforcing \emph{de facto} the NOMA paradigm on each direction.
%\item  
\end{enumerate}}

\FS{The advantage of this approach is twofold: a) the combined computational complexity of the two steps is much smaller than the complexity of the original problem, \AAB{and b) we are able to easily meet the specific NOMA constraint for user cancellation in the downlink.}}
\MM{The novelty of our work is to apply BCD, a classical algorithm for convex problems, to the solution of a highly nonconvex problem. To achieve this goal we slightly modify the original objective function by introducing a quadratic regularization term that helps to achieve the algorithm convergence.}

\FS{Moreover, even if our algorithm requires only a fraction of the complexity of the other schemes in literature, its complexity can still be prohibitive. 
Thus, to further reduce the total computational costs, we  propose a second algorithm, still based on BCD, with yet a lower complexity and similar performances. This new scheme is not iterative but employs the first two steps of the BCD algorithm to obtain the users allocation for both FD and NOMA paradigms and then recasts and solves the problem as a power allocation problem.}
Numerical simulations, which are carried out for several different scenarios, show that the proposed methods exhibits performance close to the maximum achievable optimum at a fraction of the complexity. %Moreover, our approach benefits the possibility of switching to an OMA-FD configuration per subcarrier if this improves the overall performance. 
\subsection{Outline}
The remainder of this paper is organized as follows. Section~\ref{sec:model} introduces the NOMA-FD system model; Section~\ref{sec:PF} presents the formulation of the weighted sum rate allocation problem. Section~\ref{sec:propscheme} describes the solution to the allocation problem based on the BCD algorithm and Section~\ref{sec:low-complexity} describe its low-complexity implementation. Finally, numerical results are illustrated in Section~\ref{sec:results} and conclusions are drawn in Section~\ref{sec:conclusions}.

% ===============================================
\section{System Model} \label{sec:model}
% ===============================================

We consider a NOMA-FD system, where single-antenna users are served by a single-antenna base station (BS). The available spectrum is partitioned into $F$ orthogonal fading channels collected in the set $\mathcal{F} = \{1,2,\dots,F\}$. 
We further assume that the propagation gains on each subchannel are constant for the time horizon of radio resource allocation.
We denote by $\mathcal{U}$ and $\mathcal{D}$ the sets of the $M = \left|\mathcal{U}\right|$ uplink and  $N = \left|\mathcal{D}\right|$ downlink users in the system, respectively. \MM{The set of all users is denoted as $\mc{A} = \mc{U} \cup \mc{D}$.} %In the rest of this paper we will use indexes $i$ or $n$ to refer to an user regardless of the direction of its link. On the other hand, index  $j$ will refer to uplink users only  and index $k$ exclusively to downlink users.
\MM{We denote by $s_{i}(f) \in \mathbb{C}$ the complex information symbol with unitary power of the $i$-th user on subcarrier $f$. The information symbols are assumed zero-mean independent and identically distributed (i.i.d) random variables, i.e., $\E\{{s}_{i}(f)\, {s}^{\rm H}_{n}(f')\} =1$ if $i = n$ and $f = f'$ and 0 otherwise ($s^{\rm H}$ is the complex conjugate of $s$).}
The transmitted signal for user $i$ on subcarrier $f$ is a scaled version of the information symbol with scaling factor $\sqrt{P_{i}(f)}$, so that the transmission power is $P_{i}(f)$. \MM{We stress the fact that, if $i \in \mc{U}$, $P_{i}(f)$ is the transmit power employed by user $i$ to transmit the signal towards the BS on subcarrier $f$, while, if $i \in \mc{D}$, $P_i(f)$ is the transmit power employed by the BS to transmit the signal toward user $i$ on subcarrier $f$.}

With respect to traditional systems, the FD and NOMA technologies are implemented as follows.
\begin{itemize}
\item The FD technology is implemented at the BS: by exploiting the knowledge of the transmitted downlink signal, the BS is able to cancel a large fraction of the self interference that it generates and  simultaneously transmit and receive on the same channel. On the contrary, due to hardware limitations, the mobile users operate at half duplex, i.e., they either transmit (uplink users) or receive (downlink users). 
\item Both the BS and the mobile users implement the NOMA paradigm and are able to cancel a certain number of  interfering users of the same type (uplink or downlink) on each channel. When $A$ users are allocated on a given channel, the receiver employs successive interference cancellation (SIC) to separate them. To perform SIC it is necessary that the signal to be cancelled is perfectly reconstructed  and this requires: a) perfect channel estimation, condition which we will assume true for the remainder of the paper, and b) no detection errors.%, as described in Section~\ref{sec:intCanc}.
The condition b) is obtained only if the signal-to-interference-plus-noise ratio (SINR) of  the signal that needs to be cancelled is greater at the receiver where cancellation is performed than at the intended receiver. 
\end{itemize}

Due to the high level of interference deriving from the combination of the FD and NOMA paradigms, we will set \MM{$A = 2$ for the rest of the paper, as it is is customary for SISO NOMA-FD systems \cite{Sun2017}. In case of MIMO NOMA-FD, $N$ can be much larger than 2.}
%\MM{For ease of notation, on each channel we will indicate with the term \emph{strong users} the couple of one uplink and one downlink user which, together, yield the largest weighted rate on that subcarrier, and with the term \emph{weak users} the other two NOMA users allocated on the same subcarrier. A formal definition of weighted rate will be given in Section \ref{sec:PF}. In order to maximize the system performance we assume that both for the downlink and uplink cases,  SIC is implemented by canceling the weak signal from the strong one. 
%In accordance with the strong/weak terminology,  we employ two sets of binary allocation variables for each subcarrier $f$: the allocation variables $x_{i,s}(f) \in \{0,1\}$, which are set to 1 only if the user $i \in \{\mc{U},\mc{D}\}$ is a strong user on $f$, and the allocation variables  $x_{i,w}(f) \in \{0,1\}$,   which is set to 1 if the user $i \in \{\mc{U},\mc{D}\}$ is a weak user on $f$. 
%According to the previous definitions, the set of uplink users allocated on subcarrier $f$ is $\mc{U}(f) = \set{ i \in \mc{U} \tc x_{i,s}(f) = 1 \text{ or } x_{i,w}(f) = 1}$. Similarly, the set of downlink users allocated on subcarrier $f$ is $\mc{D}(f) = \set{ i \in \mc{D} \tc x_{i,s}(f) = 1 \text{ or } x_{i,w}(f) = 1}$. For the sake of brevity, we further denote as $\mc{A}(f)$ the set of all users allocated on subcarrier $f$ i.e. $\mc{A}(f) = \mc{U}(f) \cup \mc{D}(f)$. Now, we are able to formalize the signal received for the uplink and downlink directions.}
\MM{To maximize the system performance, we assume that for both the downlink and uplink cases, SIC is implemented by one of the two users allocated on the same subcarrier, thus canceling the interference from the other one. For ease of notation, on each channel we will indicate with the term \emph{strong users} the couple of one uplink and one downlink user which cancel the interference of the other one, hereafter referred to as weak user. 
In accordance with the %strong/weak 
terminology,  we employ two sets of binary allocation variables for each subcarrier $f$: the allocation variables $x_{i,s}(f) \in \{0,1\}$, set to 1 only if the user $i \in \mc{A}$ is a strong user on $f$, and the allocation variables  $x_{i,w}(f) \in \{0,1\}$, set to 1 if the user $i \in \mc{A}$ is a weak user on $f$. 
According to the previous definitions, the set of uplink users allocated on subcarrier $f$ is $\mc{U}(f) = \set{ i \in \mc{U} \tc x_{i,s}(f) = 1 \text{ or } x_{i,w}(f) = 1}$. Similarly, the set of downlink users allocated on subcarrier $f$ is $\mc{D}(f) = \set{ i \in \mc{D} \tc x_{i,s}(f) = 1 \text{ or } x_{i,w}(f) = 1}$. %For the sake of brevity, we further denote as $\mc{A}(f)$ the set of all users allocated on subcarrier $f$, i.e. $\mc{A}(f) = \mc{U}(f) \cup \mc{D}(f)$.
We further define $\mc{A}(f) = \mc{U}(f) \cup \mc{D}(f)$.
Now, we can formalize the signal received for the uplink and downlink directions.}

\subsection{Uplink users} 
Assuming perfect cancellation of the weak users, the strong user signal is  received without any interference of other uplink users. Hence, the signal ${y}_{j}(f)$ received at the BS on sub-channel $f$ for the uplink strong user $j \in \mathcal{U}$ is
\begin{equation}
\label{eq:rxUpStrong}
\begin{aligned}
%{y}_{j}(f) &= \beta_{j}(f) \sqrt{P_{j}(f)}{s}_{i}(f) \\ 
%&+\zeta(f)\sum_{n \in \mathcal{D}(f)}\sqrt{{P}_{n}(f)}{s}_{n}(f) + \xi(f)
{y}_{j}(f) &= \beta_{j}(f) \sqrt{P_{j}(f)}{s}_{i}(f) +\zeta(f)\sum_{n \in \mathcal{D}(f)}\sqrt{{P}_{n}(f)}{s}_{n}(f) + \xi(f)
\end{aligned}
\end{equation}
where $\beta_{j}(f)$ is the  channel gain between transmitter $j \in \mathcal{U}$ and the BS, $\zeta(f)$ is the residual gain relative to the self interference (SI) term, representing the effect of non ideal cancellation at the BS,   $\mathcal{D}(f)$ is the set of downlink users allocated on subcarrier $f$, and $\xi(f)$  denotes the additive white Gaussian noise with distribution $\mathcal{CN}\left(0,\sigma^2\right)$.

The weak uplink user signal will also experience the interference caused by the other uplink user. Thus, the received signal for the uplink weak user $j' \in \mathcal{U}$ on subcarrier $f$ is
\begin{equation}
\label{eq:rxUpWeak}
\begin{aligned}
%{y}_{j'}(f)  &= \beta_{j'}(f)\sqrt{P_{j'}(f)}{s}_{j'}(f) + \beta_j(f)\sqrt{P_{j}(f)}{s}_{j}(f) \\
%&+ \zeta(f)\sum\limits_{n \in \mathcal{D}(f)}\sqrt{{P}_{n}(f)}{s}_{n}(f) + \xi(f),
{y}_{j'}(f)  &= \beta_{j'}(f)\sqrt{P_{j'}(f)}{s}_{j'}(f) + \beta_j(f)\sqrt{P_{j}(f)}{s}_{j}(f) + \zeta(f)\sum\limits_{n \in \mathcal{D}(f)}\sqrt{{P}_{n}(f)}{s}_{n}(f) + \xi(f),
\end{aligned}
\end{equation}
where user $j\in \mathcal{U}$ is the strong uplink user on channel $f$.
 
\subsection{Downlink users}
After perfect cancellation of the signal of the downlink weak user,  the received signal for the downlink strong user $k \in \mathcal{D}$ is
\begin{equation}
\label{eq:rxDownStrong}
\begin{aligned}
%{y}_{k}(f) &= \varepsilon_{k}(f) \sqrt{P_{k}(f)}{s}_{k}(f)\\
%&+\sum_{n \in \mathcal{U}(f)}\eta_{n,k}(f) \sqrt{{P}_{n}(f)}{s}_{n}(f) + \xi_{k}(f).
{y}_{k}(f) &= \varepsilon_{k}(f) \sqrt{P_{k}(f)}{s}_{k}(f) +\sum_{n \in \mathcal{U}(f)}\eta_{n,k}(f) \sqrt{{P}_{n}(f)}{s}_{n}(f) + \xi_{k}(f).
\end{aligned}
\end{equation}
where $\varepsilon_{k}(f)$ is the channel gain between the BS and receiver $k$, $\mathcal{U}(f)$ is the set of the uplink users allocated on $f$, $\eta_{n,k}(f)$ is the  channel gain between  user $n \in \mathcal{U}(f)$ and user $k$,  and   $\xi_{k,f}$  is the additive white Gaussian noise with distribution $\mathcal{CN}\left(0,\sigma^2\right)$.
Finally, the signal at the weak downlink receiver $k' \in \mathcal{D}$ is
\begin{equation}
\begin{aligned}
%{y}_{k'}(f) = &\varepsilon_{k'}(f) \sqrt{P_{k'}(f)}{s}_{k'}(f) +\varepsilon_{k}(f) \sqrt{P_{k}(f)}{s}_{k}(f)\\
%&+\sum_{n \in \mathcal{U}(f)}\eta_{n,k}(f) \sqrt{{P}_{n}(f)}{s}_{n}(f)+ \xi_{k}(f)
{y}_{k'}(f) = &\varepsilon_{k'}(f) \sqrt{P_{k'}(f)}{s}_{k'}(f) +\varepsilon_{k}(f) \sqrt{P_{k}(f)}{s}_{k}(f) +\sum_{n \in \mathcal{U}(f)}\eta_{n,k}(f) \sqrt{{P}_{n}(f)}{s}_{n}(f)+ \xi_{k}(f)
\end{aligned}
\label{eq:rxDownWeak}
\end{equation}
where user  $k\in \mathcal{D}$ is the strong downlink user on channel $f$.

\subsection{A uniform notation for the uplink and downlink channels}
A close observation of \eqref{eq:rxUpStrong}-\eqref{eq:rxDownWeak} shows that the received signal  on subcarrier $f$ for  user $i$, regardless of the fact that $i$ is uplink, downlink, strong or weak user,  can be rewritten in a more compact form as
\begin{equation}
\begin{aligned}
%{y}_{i}(f) &= h_{i,i}(f)\sqrt{P_{i}(f)}{s}_{i}(f) \\ &+ \sum_{n \in \mathcal{I}_{i,l}(f)}h_{n,i}(f)\sqrt{P_{n}(f)}{s}_{n,f} + z_{i}(f),
{y}_{i}(f) &= h_{i,i}(f)\sqrt{P_{i}(f)}{s}_{i}(f) + \sum_{n \in \mathcal{I}_{i,l}(f)}h_{n,i}(f)\sqrt{P_{n}(f)}{s}_{n,f} + z_{i}(f),
\end{aligned}
\label{eq:fusion}
\end{equation}
where the set $\mc{I}_{i}(f)$ collects the users interfering with user $i$ on subcarrier $f$, defined as
\begin{equation} \label{eq:interferenceSet}
\mc{I}_{i}(f) =
\begin{cases}
\mc{D}(f) &\quad x_{i,s}(f) = 1, \, i \in \mc{U}, \\
\mc{U}(f) &\quad x_{i,s}(f) = 1, \, i \in \mc{D}, \\
\mc{A}(f)\setminus i &\quad x_{i,w}(f) = 1,\, \forall i.
%\mc{U}(f) \cup \mc{D}(f) \setminus i &\quad x_{i,w}(f) = 1,\, \forall i.
\end{cases}
\end{equation} 
The correspondence between the coefficients $h_{n,i}(f)$ and the various propagation gains can be inferred by confronting \eqref{eq:rxUpStrong}-\eqref{eq:rxDownWeak} with \eqref{eq:fusion}. In general, if $n=i$ the  coefficient represents the direct channel between an user and the BS while if $n \neq i$ the gain coefficient is associated to an interference term. For example, when $i$ is the weak uplink user in channel $f$ described in \eqref{eq:rxUpWeak}, $h_{i,i}(f)=\beta_{i}(f)$ and $h_{n,i}(f)$ is either $\beta_{n}(f)$ if $n$ is the strong uplink user or $\zeta(f)$ if $n\in\mathcal{D}(f)$. The term $z_{i}(f)$ represents the noise in the uplink or downlink receiver  with variance $\sigma^{2}=\E\left\{|z_{i,f}|^2\right\}$.
Fig.~\ref{fig:scenario} depicts a toy example of four NOMA-FD users transmitting on the same channel. 
\begin{figure}[tbh]
\centering
\includegraphics[width= 8cm]{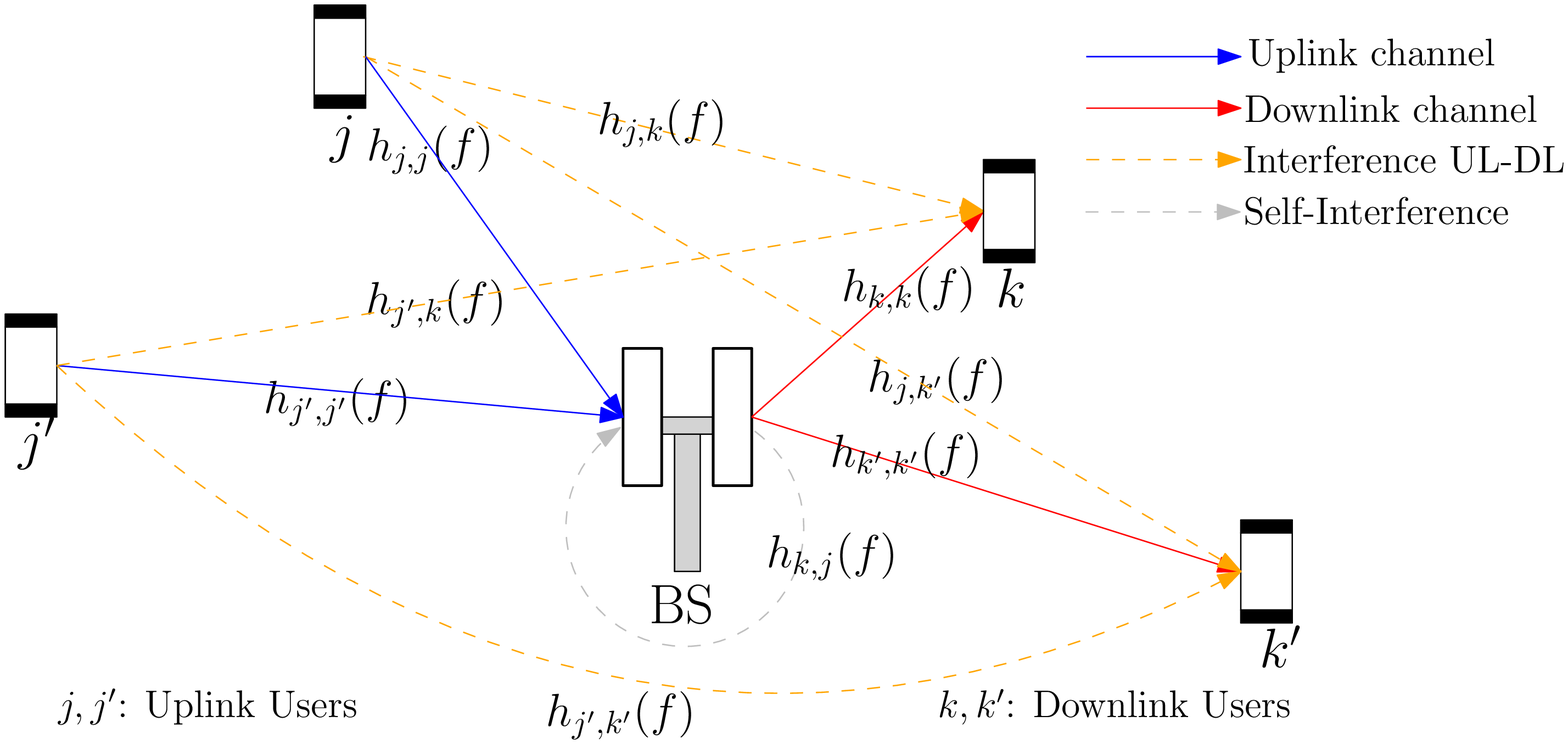}
\caption{Toy example of four users transmitting on the same channel.}
\label{fig:scenario}
\end{figure}

According with the formulation provided in~\eqref{eq:fusion}, the achievable normalized rate  of a user $i$ allocated on subcarrier $f$ % for both uplink and downlink directions is
is
\begin{equation}
R_{i}(f)= \log_{2}\left(1 + \gamma_{i}(f) \right)
\label{eq:Rate}
\end{equation}
%with $l \in \{w,s\}$ and 
with the signal-to-interference-plus-noise ratio (SINR) is 
\begin{equation}
\gamma_{i}(f)=\frac{\left|h_{i,i}(f)\right|^2 P_{i}(f)}
{\displaystyle\sum_{n \in  \mathcal{I}_{i}(f)} \left|h_{n,i}(f)\right|^2 P_{n}(f) + \sigma^2 }.
\end{equation}

\subsection{Interference cancellation}
\label{sec:intCanc}
Interference cancellation is a key task in a NOMA-FD system and in our context requires special scrutiny. In the \emph{uplink} direction, interference cancellation is simple:  all signals are superimposed on the same channel and the BS decides the decoding order without any particular constraint, so that in our setting we can decide to estimate first and cancel the uplink weak user in order to  receive the strong user without any uplink interference, as in \eqref{eq:rxUpStrong}.
In the \emph{downlink}, the  data stream intended for the weak user's receiver needs to be canceled at the strong user's receiver. Accordingly, it is necessary that the weak users's signal has a higher SINR at the strong user receiver than at its own receiver. Unfortunately, the signal experiences different channel attenuation and interference at the two receivers. Thus, to illustrate this new constraint, let us focus on the downlink channel $f$, where  $k\in\mc{D}$ is the strong user and $k'\in\mc{D}$ is the weak user. User $k'$ transmits with rate $R_{k'}(f)$ as defined in \eqref{eq:Rate}. The achievable rate of user $k'$ at the receiver of $k$  is $R_{k',k}(f)=\log_{2}\left(1 + \gamma_{k',k}(f) \right)$, where it is
\begin{equation}
\label{eq:rateWS}
\gamma_{k',k}(f)=\frac{\left|h_{k,k}(f)\right|^2 P_{k'}(f)}{\left|h_{k,k}(f)\right|^2 P_{k}(f) + \hspace{-3pt} \sum\limits_{n \in  \mathcal{U}(f)} \left|h_{n,k}(f)\right|^2 P_{n}(f) + \sigma^2 }.
\end{equation}
Accordingly, the condition for perfect cancellation of the weak user $k'$ is% that $R_{k',k}(f)$ is larger than $R_{k'}(f)$, i.e.,
\begin{equation}
\label{eq:WSWW}
R_{k',k}(f) \geq R_{k'}(f), \quad  
k,k' \in \mc{D} \,\Big|\, x_{k,s}(f) = x_{k',w}(f) = 1,
%\begin{aligned}
%x_{k,s}(f) &= 1, \\
%x_{k',w}(f) &= 1, 
%\end{aligned}
\end{equation}
The  inequality~\eqref{eq:WSWW} is non-linear and non-convex in the users' power but can be converted into a linear one. By exploiting the monotone property of the $\log$ function, one obtains the equivalent form of \eqref{eq:WSWW} as $\gamma_{k',k}(f) \geq \gamma_{k'}(f),$
%
%\begin{equation}
%\label{eq:WSWW2}
%\gamma_{k',k}(f) \geq \gamma_{k'}(f),
%\end{equation}
%
which, provided that $P_{k'}(f) > 0$, is equivalent to 
\begin{equation}
\label{eq:GAMMA}
\begin{aligned}
\Gamma_{k',k}(f) =& \sum_{j \in \mc{U}(f)} \theta_{j,f}^{(k',k)} P_{j}(f) +\delta_f^{(k',k)} \geq 0.
\end{aligned}
\end{equation}
where
\begin{equation}
\label{eq:thetaDef}
\begin{aligned}
\theta_{j,f}^{(k',k)} &= |h_{k,k}(f) \, h_{j,k'}(f)|^2 - |h_{k',k'}(f) \, h_{j,k}(f)|^2, \quad
\delta_f^{(k',k)} &= \sigma^2 \left(|h_{k,k}(f)|^2 - |h_{k',k'}(f)|^2 \right).
\end{aligned}
\end{equation}
%with $k,k' \in \mc{D}, \, x_{j,k}^{(s)}(f) = x_{j,k'}^{(w)}(f) = 1$, $\forall j \in \mc{U}$.
%
An interesting feature of \eqref{eq:GAMMA} is that it depends only on the channel gains of users $k$ and $k'$ and not on their power. Nevertheless, it is worth noting that indeed there is a dependence on the users' power: the relation~\eqref{eq:GAMMA} is binding only if $P_{k}(f)$ and $P_{k'}(f)$ are both different from zero. Otherwise, even if only one of the two allocated users  is transmitted with zero power, the NOMA paradigm is not implemented on subcarrier $f$ for downlink transmission, and the constraint does not hold anymore.

\section{Problem Formulation}
\label{sec:PF}
\MM{Let us denote by $\mb{x}_s = \{x_{i,s}(f),\, i\in\mc{A}, f\in\mc{F}\}$ and by $\mb{x}_w = \{x_{i,w}(f),\, i\in\mc{A}, f\in\mc{F}\}$ the vectors collecting the strong and weak allocation variables, respectively, and by $\mb{x}$ the vector stacking $ \mb{x}_s$ and $\mb{x}_w$.
%\FS{Let us denote by $\mb{x} = \{x_{i,m}(f),\, i\in\mc{A},\, m \in \{w,s\},\, f\in\mc{F}\}$ the vectors collecting the allocation variables.}
Moreover, let us denote by $\mb{P}(f) = \set{P_{i}(f),\, i \in \mc{A}}$ the vector collecting all the powers for all users on subcarrier $f$ and by $\mb{P} = \set{\mb{P}(f),\,\forall f \in \mc{F}}$ the vector of all powers for all the users and subcarriers.}
\MM{Given the  the weights $\alpha_{i}$, $i \in \mc{A}$, employed to enforce a certain degree of fairness among users, the weighted sum rate on subcarrier $f$ can be computed for  the strong and weak couple of users as
\begin{equation}
\label{eq:utilitySW}
\begin{aligned}
U_s(\mb{x}_s, \mb{P}(f)) &= \sum_{i \in \mc{A}(f)}x_{i,s}(f) \alpha_i  R_i(f), \qquad
U_w(\mb{x}_w, \mb{P}(f)) &= \sum_{i \in \mc{A}(f)}  x_{i,w}(f) \alpha_iR_i(f).
\end{aligned}
\end{equation}
We are now able to define the utility function $U(\mb{x}, \mb{P})$ as the weighted sum rate of all users on all sucarriers, i.e.
\begin{equation}
\label{eq:utility}
\begin{aligned}
U(\mb{x}, \mb{P}) = \sum_{f \in \mc{F}} U_s(\mb{x}_s, \mb{P}(f)) + U_w(\mb{x}_w, \mb{P}(f))
\end{aligned}
\end{equation}}

The maximum weighted sum rate allocation problem can be formulated as 
\begin{align}  \label{PB:RA}
&\underset{\substack{ \mathbf{P} \succeq 0, \mb{x}} }{\max}
          ~ U\left(\mb{x}, \mb{P}\right) \\
    & \quad \text{subject to} \notag \\ 
        &  \qquad \sum_{f \in \mc{F} } \AAB{\big(x_{i,s}(f) + x_{i,w}(f) \big)} P_{i}(f) \le P_{U}, \quad  \forall i \in \mathcal{U} \tag{\ref{PB:RA}.a}\label{PB:PU} \nonumber\\
        & \qquad  \sum_{f \in \mc{F} } \sum_{i \in \mathcal{D}} \AAB{\big(x_{i,s}(f) + x_{i,w}(f)\big)}  P_{i}(f) \le P_{D} \tag{\ref{PB:RA}.b}\label{PB:PD}\\
        &  \qquad \sum_{i \in \mathcal{U}} x_{i,s}(f) \le 1, \quad \forall f \tag{\ref{PB:RA}.c}\label{PB:KU1} \nonumber\\
        &  \qquad \sum\limits_{i \in \mathcal{U}} x_{i,w}(f) \le 1, \quad \forall f \tag{\ref{PB:RA}.d}\label{PB:KU2} \nonumber\\
        &  \qquad \sum\limits_{i \in \mathcal{D}} x_{i,s}(f) \le 1, \quad \forall f \tag{\ref{PB:RA}.e}\label{PB:KD1}  \nonumber\\
        &  \qquad \sum\limits_{i \in \mathcal{D}} x_{i,w}(f) \le 1, \quad \forall f \tag{\ref{PB:RA}.f}\label{PB:KD2}  \nonumber\\
        &  \qquad x_{i,s}(f)+x_{i,w}(f) \le 1, \quad  \forall i \in \{\mathcal{U},\mc{D}\}, \quad\forall f \tag{\ref{PB:RA}.g}\label{PB:KU11} \nonumber\\
        &  \qquad \Gamma_{k',k}(f) \geq 0,\, \begin{aligned} \forall k',k &\in \mc{D} \\ \forall f &\in \mc{F} \end{aligned} \,\Big|\, \begin{aligned} &x_{k',w}(f) = x_{k,s}(f) = 1,\\
        &P_{k'}(f)>0, P_{k}(f)>0. \end{aligned}  \tag{\ref{PB:RA}.i}\label{PB:WSWW}
\end{align}
Constraints~\eqref{PB:PU} and~\eqref{PB:PD} represent the power budget for the uplink and downlink users.
Constraints~\eqref{PB:KU1}-\eqref{PB:KU11} ensure that a maximum of only one couple of strong and weak users for both uplink and downlink directions can be allocated for each channel. Finally, constraints~\eqref{PB:WSWW} guarantee  successful cancellation of the downlink weak  user.  The constraints formally explicit that both powers of the  allocated downlink  users must be greater than zero to be binding. \MM{Going along with other papers dealing with similar systems \cite{Sun2017,Xiao17}, we have opted for weighting the rate of each user by user user-specific coefficients, rather then setting minimum data rate requirements. This choice is in line with the handling of data traffic typical of NOMA-FD systems in 5G.}

The weighted sum rate problem~\eqref{PB:RA} is not convex because of the structure of  the utility function in $\mb{P}$ and $\mb{x}$ and because of the binary integer nature of the allocation variable $\mb{x}$. Moreover, since the size of both $\mb{P}$ and $\mb{x}$ is $F M^2 N^2$, solving~\eqref{PB:RA} is further made complex by the large number of variables and constraints involved in the optimization process.
\MM{Table \ref{tab:notation} summarizes  most of the notation  employed in this paper.}
{\begin{table}[thb]
\smaller
\centering
\begin{tabular}{p{0.18\textwidth}|p{0.72\textwidth}}
%\begin{tabular}{p{0.15\textwidth}|p{0.29\textwidth}}
\toprule
%\multicolumn{2}{c}{Notation used} \\
%\midrule
$\mc{U}, \mc{D}, \mc{A}$, $\mc{F}$ & set of $M$ uplink, $N$ downlink, all users, and $F$ sucarriers \\
$\mc{U}(f), \mc{D}(f), \mc{A}(f)$ & set of uplink, downlink, all users allocated on $f \in \mc{F}$ \\
\midrule
$\mb{x}$ & vector collecting all the allocation variables \\
$\mb{x}_{s(w)}$ & vector collection strong (weak) allocation variables\\
$\mc{X}_{s,(w)}^{(\ell)} $ & feasible set for strong (weak) allocation variables on iteration $\ell$ \\
\midrule
$\mb{P}$ & vector of all power coefficients \\
$\mb{P}(f)$ & vector of power coefficient of $f \in \mc{F}$ \\
$\mb{P}_{s,(w)}^{(\ell)}$ & vector of power coefficient for strong (weak) users allocated on iteration $(\ell)$  \\
$\mb{P}_{U(D)}$ & vector of power coefficient for uplink (downlink) users \\
$\mc{P}_{s,(w)}^{(\ell)} $ & feasible set for strong (weak) power coefficients on iteration $\ell$ \\
$P_{U(D)}$ & maximum power available for uplink (downlink) users \\
\midrule
$h_{n,i}(f)$ & channel gain coefficient between $n \in \mc{A}$ to $i \in\mc{A}$\\
$\mc{I}_i(f)$ & set of users interfering with $i \in \mc{A}$ defined in~\eqref{eq:interferenceSet} \\
$\mc{C}_i(f)$ & set of users which are interfered by $i \in \mc{A}$ defined in~\eqref{eq:intSetC} \\
$\gamma_i(f)$, $R_i(f)$ & SINR and achievable rate of $i \in \mc{A}$ on $f\in\mc{F}$ \\
$\gamma_{k',k}(f), R_{k',k}(f)$ & SINR and achievable rate of weak $k' \in \mc{D}$ evaluated at strong user $k \in \mc{D}$ for $f\in\mc{F}$ \\
$\Gamma_{k',k}(f)$ & linear NOMA constraint~\eqref{eq:GAMMA} to be fulfilled when weak and strong users $k',k\in\mc{D}$ are allocated on $f\in\mc{F}$\\
\midrule
$U_{s(w)}(\mb{x}_{s(w)}, \mb{P}(f))$ &  weighted sum of the strong (weak) users on $f\in\mc{F}$~\eqref{eq:utilitySW} \\
$U(\mb{x},\mb{P})$ & overall weighted sum rate~\eqref{eq:utility} \\
\midrule
$g_{s(w)}$ & Lagrangian dual function of the strong (weak) allocation process \\
$L_{j,k,s(w)}(\mb{P}(f), \bm{\mu})$ & Lagrangian auxiliary function accounting for the strong (weak) allocation $j\in\mc{U},k\in\mc{D}$ on $f\in\mc{F}$\\
$L_\text{cav}, L_\text{vex}$ & concave and convex part of the Lagrangian auxiliary function; defined in~\eqref{eq:dc} \\
\bottomrule
\end{tabular}
\caption{Notation used throughout the paper.}
\label{tab:notation}
\end{table}}

\MM{\section{A globally convergent algorithm for \eqref{PB:RA}} 
\label{sec:propscheme}
Recent research~\cite{Xu2017} has showed that \emph{block coordinate descent} (BCD), a class of iterative  algorithms traditionally employed to solve convex problems,  can be used to address large-scale nonconvex optimization problems. The idea behind coordinate descent algorithms is to find a solution  by successively solving the problem along a subset of optimization variables, keeping all the other variables fixed. Under certain conditions, the block coordinate approach yields optimal or quasi-optimal results with a fraction of the original complexity.  When the problem is noncovex, the goal is global convergence to a local optimum, hopefully very close to the optimal solution.}

\FS{Owing to the particular structure of the system considered, we propose a BCD algorithm that alternates iteratively between the solution of the allocation problem for the strong and for the weak users.}
\MM{The update rule at iteration $\ell$ is
\begin{equation}
\label{eq:updateRule}
\left\{ 
%\begin{array}{*{20}{l}}
\begin{aligned}
\mb{x}_{s}^{(\ell)},\mb{P}_{s}^{(\ell)} =\argmax\limits_{\substack{\mb{x}\in\mc{X}_{s}^{(\ell)}, \mb{P}\succeq 0}} 
\Big\{& U\left(\mb{x},\mb{P};\mb{x}_{w}^{(\ell-1)},\mb{P}_{w}^{(\ell-1)}\right) 
- K^{(\ell)} || \mb{P} - \mb{P}_s^{(\ell-1)}||^2 \Big\}, \\
\mb{x}_{w}^{(\ell)},\mb{P}_{w}^{(\ell)} =\argmax\limits_{\substack{\mb{x}\in\mc{X}_{w}^{(\ell)}, \mb{P}\succeq 0}} 
\Big\{& U \left(\mb{x}_{s}^{(\ell)},\mb{P}_{s}^{(\ell)};\mb{x},\mb{P}\right)
- K^{(\ell)}|| \mb{P} - \mb{P}_w^{(\ell-1)}||^2\Big\},\\
\text{subject to } \eqref{PB:PU},~\eqref{PB:PD},&~\eqref{PB:WSWW},
\end{aligned} 
\right.
\end{equation}
where we have explicitly expressed $U$ as a function of the variables currently under maximization and the variables obtained by the previous step.}

\MM{Each iteration of the BCD algorithm is composed by two steps: in the first step the objective function is maximized  with respect to the strong users' block of variables, having fixed the weak users variables; in the second step the utility is maximized with respect to the weak users, having fixed the strong users' variables. The algorithm is initialized by setting all the weak allocation variables to zero in the first step of iteration $\ell=0$, i.e., at $\ell= 0$ all users are potentially strong users. Moreover, the initial power coefficients  are all set to $0$.}

\MM{In \eqref{eq:updateRule} the allocation variable constraints of problem~\eqref{PB:RA} are collected in the following feasible sets
\begin{equation}
\begin{aligned}
\label{eq:setX}
\mc{X}_{s} \hspace{-2.5pt}&= \hspace{-2.5pt}\set{x_{i,s}(f)\in\{0,1\},\, i \in\{\mathcal{U},\mc{D}\} \tc x_{i,w}(f)=0}
\cap \Big\{ \sum_{i \in \mc{U}} x_{i,s}(f)\le 1, \, \sum_{i \in \mc{D}} x_{i,s}(f)\le 1, \, \forall f\Big\},
\\
\mc{X}_w \hspace{-2.5pt}&= \hspace{-2.5pt}\set{x_{i,w}(f)\in\{0,1\},\, i \in\{\mathcal{U},\mc{D}\}\tc x_{i,s}(f)=0} 
\cap \Big\{ \sum_{i \in \mc{U}} x_{i,w}(f)\le 1, \, \sum_{i \in \mc{D}} x_{i,w}(f)\le 1, \, \forall f\Big\}.
\end{aligned}
\end{equation}
It is worth noting that the composition of the sets $\mc{X}_s$ and $\mc{X}_w$ may change from iteration to iteration, because the condition of   being   the strong or weak user on a given channel greatly depends on the presence of the interference.}

\MM{As illustrated in Sections~\ref{sec:SA} and~\ref{sec:WA}, the allocation problem for the two steps in~\eqref{eq:updateRule} can be solved to a local optimum, so that, in theory, the  solution of a new step of an iteration  can only improve the sum rate with respect to the previous one.}

\MM{To enforce the convergence of the algorithm, the objective function in~\eqref{eq:updateRule}  has  been slightly  modified with respect to~\eqref{PB:RA}  by adding the regularization term $K^{(\ell)} || \mb{P} - \mb{P}_m^{(\ell -1)}||^2$, with $m \in \{w,s\}$. If the problem were concave, there would be only one maximum and the value of objective function $U$ would grow for each step of the algorithm until convergence.} %In this case, due to the non convexity of $\eqref{PB:RA}$, the optimization procedure for solving each step of \eqref{eq:updateRule} might end up in a different local optimum, which could have a lower value than the starting point.
\FS{However, a non concave function may have a larger number of local optima and an optimization procedure without regularization might end up in a different local optimum, which could have a lower value than the starting point. Nevertheless, as shown in~\cite{Xu2017}, provided that the search is performed in a sufficiently small interval, the function can be approximated locally as concave.}
\MM{Let $U\left(\mb{x}^{(\ell-1)},\mb{P}^{(\ell-1)}\right)$ be  the value of the utility at the end of iteration $\ell-1$  of \eqref{eq:updateRule}   and  $\mb{x}_{s}^{(\ell)},\mb{P}_{s}^{(\ell)}$ be the solution of the first step of iteration $\ell$. Given the  interference determined by the allocations $\mb{x}_{w}^{(\ell-1)},\mb{P}_{w}^{(\ell-1)}$,  the solver maximizes the utility function $U$ and, if $\mb{P}_{s}^{(\ell)}$  lies in the proximity of $\mb{P}_{s}^{(\ell-1)}$, it is}
\begin{equation}
\label{eq: algImpr}
U\left(\mb{x}_{s}^{(\ell)},\mb{P}_{s}^{(\ell)},\mb{x}_{w}^{(\ell-1)},\mb{P}_{w}^{(\ell-1)}\right)\ge U\left(\mb{x}^{(\ell-1)},\mb{P}^{(\ell-1)}\right).
\end{equation}
%
%Unfortunately, since the problem is not concave, the algorithm may end up in an optimum which is not in the proximity of $\mb{P}_{s}^{(\ell-1)}$ and accordingly inequality \eqref{eq: algImpr} might not be true anymore.  
\MM{By adding the regularization term we introduce a penalty for any solution that is far away from $\mb{P}_{s}^{(\ell-1)}$, so that $\mb{P}_{s}^{(\ell)}$ is either forced to be in the proximity of $\mb{P}_{s}^{(\ell-1)}$ or to exhibit  a gain that overcomes the  penalty of the regularization factor $K^{(\ell)} || \mb{P}_{s}^{(\ell)} - \mb{P}_s^{(\ell -1)}||^2$.  In both cases, the utility function will improve with respect to its previous value as in \eqref{eq: algImpr}. The same reasoning applies also to the second step of \eqref{eq:updateRule}. The value of $K^{(\ell)}$ is chosen following the method in ~\cite{Xu2017}.}
%To enforce the condition on the proximity of the solution, it is sufficient to consider only the vector of continuous variable $\mb{P}$, without any conditions on $\mb{x}$. In fact, when a channel is not assigned to a given user than the power is set to 0 and the vector  $\mb{P}$ encompasses also the allocation information.}

\begin{theorem}
\label{theo:convBCD}
The BCD algorithm converges to a local optimum of \eqref{PB:RA}.
\end{theorem}
\begin{proof}
\MM{We have just shown that at each step of \eqref{eq:updateRule} the objective function either increases or  does not change. The convergence of the BCD algorithm is guaranteed by the fact that the weighted sum rate is bounded and thus it can not increase indefinitely but at certain point must converge on a local optimum. At convergence,  the regolarization term $K^{(\ell)} || \mb{P} - \mb{P}_m^{(\ell -1)}||^2$ will be null and the objective function of  \eqref{eq:updateRule} will coincide with the objective function of \eqref{PB:RA}.}
\end{proof}
\FS{In the  following two subsections we will address both the strong and weak allocation problems of~\eqref{eq:updateRule}.} %Because of the mixed integer and non-convex nature of, we have studied an ad hoc solution in the Lagrangian dual domain. }

\subsection{Strong users Allocation (SA)}
\label{sec:SA}
\MM{Let us focus on the allocation problem for the strong users at iteration $\ell$.}
% During the previous weak allocation $\ell-1$, power coefficients $\mb{P}_w^{(\ell-1)}$ and allocation variables $\mb{x}_w^{(\ell-1)}$ have been evaluated.
%Accordingly, our aim is solving the problem}
Taking into account the power allocated in the weak allocation step at iteration $\ell-1$, the remaining power budget for uplink users and BS is 
\begin{equation}
\label{eq:remPowS}
\begin{aligned}
P_{U,i,s}^{(\ell)} &= P_{U} - \sum\limits_{f\in\mc{F}} x_{i,w}^{(\ell-1)}(f) P_{i}^{(\ell -1)}(f), \qquad
P_{D,s}^{(\ell)} &= P_{D} - \sum\limits_{f \in \mc{F}}\sum\limits_{i \in \mathcal{D}} x_{i,w}^{(\ell-1)}(f) {P}_{i}^{(\ell-1)}(f),
\end{aligned}
\end{equation}
and the strong user problem of \eqref{eq:updateRule} is formulated as 
\begin{align}
\label{PB:SA}
&\max\limits_{\substack{\mb{x}\in\mc{X}_{s}^{(\ell)}\\ \mb{P}\succeq 0}} U\left(\mb{x},\mb{P};\mb{x}_{w}^{(\ell-1)},\mb{P}_{w}^{(\ell-1)}\right) - K^{(\ell)} || \mb{P} - \mb{P}_s^{(\ell -1)}||^2 \\
& \quad \text{subject to} \notag \\
& \qquad \sum\limits_{f \in \mc{F} } x_{i,s}(f) P_{i}(f) \le P_{U,i,s}^{(\ell)}, \quad  \forall i \in \mathcal{U} \tag{\ref{PB:SA}.a}\label{PB:PUstrong} \nonumber\\
& \qquad  \sum\limits_{f \in\mc{F}} \sum\limits_{i \in \mathcal{D}} x_{i,s}(f){P}_{i}(f) \le P_{D,s}^{(\ell)}, \tag{\ref{PB:SA}.b}\label{PB:PDstrong} \\
&  \qquad \Gamma_{k',k}(f) \geq 0,\, \begin{aligned} \forall k',k &\in \mc{D} \\ \forall f &\in \mc{F} \end{aligned} \,\Big|\, \begin{aligned} &x_{k',w}(f) = x_{k,s}(f) = 1,\\
        &P_{k'}^{(\ell-1)}(f)>0. \end{aligned}  \tag{\ref{PB:SA}.c}\label{PB:WSWWstrong}
\end{align}
%
%In order to reduce the problem to a tractable form, we have relaxed the binding condition $P_{k}(f) > 0$ of constraint~\eqref{PB:WSWW}, for the strong user $k \in \mc{D}$. \AAB{This formulation might lead to a still feasible but sub-optimal solution: when the power of the strong downlink user $k$ is 0 on subcarrier $f$, uplink powers will be constrained so that $\Gamma_{k',k}(f) \geq 0$, even if the NOMA paradigm is not implemented in the downlink. However, this formalization helps to greatly reduce the complexity of the problem, as described in Appendix~\ref{sec:powerAllocation}.}

Even if not convex, problem~\eqref{PB:SA} can be efficiently solved in the \emph{Lagrangian dual domain}. These types of problems present several local optima, due to the non-convexity of the objective function. Nevertheless, it has been shown in~\cite{Abrardo2017} that by choosing a proper set of initial values, even if strong duality does not hold, a local solution very close the global optimum of the primal problem~\eqref{PB:SA} can be found by solving the  Lagrangian dual problem
\begin{equation} 
\label{eq:LagDPStrong}
\min_{\bm{\mu}\succeq 0} g_s(\bm{\mu}),
\end{equation}
where $\bm{\mu} = [\mu_0, \mu_1, \dots \mu_M]^T \in \mathbb{R}_{+}^{(U+1)}$ is the vector of the Lagrangian dual variables  associated to the constraints~\eqref{PB:PDstrong} for $\mu_0$, and~\eqref{PB:PUstrong} for $\mu_i$, $i \in 1, \dots M$.
The Lagrangian dual function $g_s(\bm{\mu})$ is
\begin{equation} 
\label{eq:LagStrong1}
\begin{aligned}
&g_s(\bm{\mu})  = \hspace*{-2pt}  \max_{\substack{\mb{x} \in \mc{X}_s^{(\ell)}\\\mathbf{P} \in \mc{P}_{s}^{(\ell)}}} \hspace*{-3pt} U\Big(\mb{x},\mb{P};\mb{x}_{w}^{(\ell-1)},\mb{P}_{w}^{(\ell-1)}\Big) \hspace*{-3pt}-\hspace*{-3pt} K^{(\ell)} || \mb{P} - \mb{P}_s^{(\ell -1)}||^2  \\ 
&+ \sum_{j \in \mc{U}} \mu_j  \big( P_{U,i,s}^{(\ell)} - \hspace{-2pt}\sum_{f\in\mc{F}} P_{j}(f) \big) + \mu_{0} \big(P_{D,s}^{(\ell)} - \sum_{f\in\mc{F}} \sum_{k \in \mc{D}} P_{k}(f) \big),
\end{aligned}
\end{equation}
where
%\begin{equation} \label{eq:feasibleStrong}
%\begin{aligned}
%\mc{P}_s^{(\ell)} = \big\{\mb{P} \tc 0 \le &P_j(f) \le P_U^{(\ell)}, \forall j \in \mc{U},\\   0 \le &P_{k}(f) \le P_D^{(\ell)}, \forall k \in \mc{D}\big\} \bigcap \\
%\big\{P_k'(f)\Gamma_{k',k} \ge 0&, \, \forall k',k \tc x_{k',w}(f) = x_{k,s}(f) = 1,\, f \in \mc{F}\big\}.
%\end{aligned}
%\end{equation}
\begin{equation} \label{eq:feasibleStrong}
\begin{aligned}
%\mc{P}_{s}^{(\ell)} = \Big\{&\mb{P} \tc \mb{P} \succeq 0,\, P_{k'}^{(\ell-1)}(f) \Gamma_{k',k}(f) \geq 0, \\ &\forall k',k \in \mc{D} \tc x_{k',w}(f) = x_{k,s}(f) = 1,\,  \forall f \Big\}.
\mc{P}_{s}^{(\ell)} = \Big\{&\mb{P} \tc \mb{P} \succeq 0,\, P_{k'}^{(\ell-1)}(f) \Gamma_{k',k}(f) \geq 0, \,\forall k',k \in \mc{D} \tc x_{k',w}(f) = x_{k,s}(f) = 1,\,  \forall f \Big\}.
\end{aligned}
\end{equation}
Now, we define the Lagrangian auxiliary function which accounts for the contribution of the strong uplink-downlink allocation ($j,k$) on subcarrier $f$ as
\begin{equation}
\label{eq:auxLagStrong}
\begin{aligned}
L_{j,k,s}&(\mb{P}(f),\bm{\mu}) = \alpha_j R_j(f) + \alpha_k R_k(f) + U_w(\mb{x}_w^{(\ell -1)}, \mb{P}(f))\\
 &- K_\ell ||[P_j(f), P_k(f)] - [P_j^{(\ell-1)} (f),P_k^{(\ell-1)}(f)]||^2 
- \mu_j P_j(f) - \mu_0 P_{k}(f).
\end{aligned}
\end{equation}
Hence, neglecting  the terms not involved in the optimization, the dual Lagrangian function~\eqref{eq:LagStrong1} is equivalent to the following
\begin{equation}
\label{eq:lagComp}
\begin{aligned}
g_s(\bm{\mu}) = &\sum_{f\in\mc{F}}\hspace{-3pt}  \max_{\substack{\mb{x} \in \mc{X}_s^{(\ell)} \\ \mathbf{P}(f) \in \mc{P}_s^{(\ell)}}} \sum_{j \in \mc{U}} \sum_{k \in \mc{D}} x_{j,s}(f) x_{k,s}(f) L_{j,k,s}(\mb{P}(f),\bm{\mu}).\\
%&+\mu_{0} \rho_D P_D+ \sum_{j \in \mc{U}} \mu_j   \rho_U P_U,
\end{aligned}
\end{equation}
It is now possible to  solve the maximization problem in \eqref{eq:lagComp} by computing first  the power coefficients and then the allocation variables.

\FS{The power coefficients can be evaluated by maximizing the auxiliary Lagrangian function~\eqref{eq:auxLagStrong}. Keeping fixed a single strong couple $j \in \mc{U}$, $k \in \mc{D}$, and the previously selected weak users by means of $\mb{x}_w^{(\ell-1)}$ and $\mb{P}^{(\ell-1)}_w$, the maximization process on $\mb{P}(f)$ involves $P_j(f)$ and $P_k(f)$ only. Hence, the power coefficient can be obtain solving
\begin{equation} \label{eq:maxLstrong}
\begin{aligned}
\left[P^{*}_{j}(f),P^{*}_{k}(f)\right]=\argmax_{\mb{P}(f) \in \mc{P}_s^{(\ell)}} & \, L_{j,k,s}(\mb{P}(f),\bm{\mu}), \quad 
\begin{aligned}
j &\in \mc{U}, \\
k &\in \mc{D},
\end{aligned}
\end{aligned}
\end{equation}
for \emph{each possible couple of strong users} $j \in \mc{U}$ and $k \in \mc{D}$.}
Problem~\eqref{eq:maxLstrong} is a \emph{difference-of-convex} (\emph{DC}) function maximized on a convex set, thus it can be solved by state-of-the-art methods as described in Appendix~\ref{sec:powerAllocation}.

Once the  power coefficients are computed for every possible couple of users, the optimal allocation variables for each subcarrier are obtained by selecting the couple which maximizes the Lagrangian function, i.e.,
\begin{equation} \label{eq:alloVarStrong}
x_{j(f),s}(f) = 1, x_{k(f),s}(f) = 1, \quad \forall f \in\mc{F},
\end{equation}
where $j(f)\in \mc{U}$ and $k(f) \in \mc{D}$ and it is 
\begin{equation}
[j(f), k(f)]= \arg\max_{j,k} L_{j,k,s}(\mb{P}^*(f),\bm{\mu}), \,\, f \in \mc{F}.
\end{equation}
\FS{All the allocation variable and the power coefficient of users not allocated are set to 0, i.e.
%as strong are set to 0, i.e. $\forall f \in \mc{F}$
\begin{equation}
\begin{aligned}
x_{j,s}(f) = 0, \quad& P_j(f) = 0,  \\
x_{k,s}(f) = 0, \quad& P_k(f) = 0,
\end{aligned} \quad
\begin{aligned}
\forall j \neq j(f) \in \mc{U}, \\
\forall k \neq k(f) \in \mc{D}.
\end{aligned}
\quad \forall f \in \mc{F}
\end{equation}}
Having found the power coefficients and the allocation variables for a given value of $\bm{\mu}$, the algorithm computes the subgradient of~\eqref{eq:LagStrong1} as described in Appendix~\ref{sec:dual}, and it updates the Lagrangian dual variables with the \emph{cutting plane ellipsoid} method~\cite{Yu2006}, iterating until convergence.

The Strong user Allocation (SA) algorithm is summarized in Algorithm~\ref{alg:SA}. The main complexity burden is represented by computing  $F M N$ metrics at each iteration on $\bm{\mu}^{(t)}$.
At  convergence, the algorithm yields to local optimum variables, collected in vectors $\mb{x}_s^{(\ell)}$ and $\mb{P}_{s}^{(\ell)}$.

\begin{algorithm}
\footnotesize
\caption{Strong user Allocation (SA)}
\label{alg:SA}
\textbf{Input:}
$\mathbf{P} \leftarrow [\mb{P}_s^{(\ell-1)}, \mb{P}_w^{(\ell-1)}] $, $t \leftarrow 1$, $\Delta \leftarrow 1$: \textbf{Output:} $\mb{P}_s^{(\ell)}$, $\mb{x}_s^{(\ell)}$\;
\While{$\Delta > 0$}
{
	\For {$f \in \mc{F}$}
	{
		\For {$(j,k) \in (\mathcal{U},\mathcal{D})$}
		{			
			Compute $P_{j}(f)$ and $P_{k}(f)$ by solving~\eqref{eq:maxLstrong};
		}
		Compute $x_{i,s}(f)$ as in~\eqref{eq:alloVarStrong};\\		
	}
	Update $\bm{\mu}^{(t)}$ as described in Appedix~\ref{sec:dual}\;
	$\Delta \leftarrow \left\Vert \bm{\mu}^{(t)} - \bm{\mu}^{(t-1)}\right\Vert$\;
	$t \leftarrow t + 1$\;
}
\end{algorithm}
\normalsize

\vspace{-0.5cm}
\subsection{Weak users Allocation (WA)} 
\label{sec:WA}

\MM{After the strong users are allocated as described above, power coefficients $\mb{P}_s^{(\ell)}$ and allocation variables $\mb{x}_s^{(\ell)}$ are obtained. In analogy with \eqref{eq:remPowS},  the remaining power budget for weak uplink users and BS at iteration $\ell$ is}
\begin{equation}
\begin{aligned}
P_{U,i,w}^{(\ell)} &= P_{U} - \sum\limits_{f\in\mc{F}} x_{i,s}^{(\ell)}(f) P_{i}^{(\ell)}(f), \qquad
P_{D,w}^{(\ell)} &= P_{D} - \sum\limits_{f \in \mc{F}}\sum\limits_{i \in \mathcal{D}} x_{i,s}^{(\ell)}(f) {P}_{i}^{(\ell)}(f),
\end{aligned}
\end{equation}
and the weak users' allocation problem at iteration $\ell$ is
\begin{align}
\label{PB:WA}
&\max\limits_{\substack{\mb{x}\in\mc{X}_{w}^{(\ell)}\\ \mb{P}\succeq 0}} U\left(\mb{x}_s^{(\ell)},\mb{P}_s^{(\ell)};\mb{x},\mb{P}\right) - K^{(\ell)} || \mb{P} - \mb{P}_w^{(\ell -1)}||^2 \\
& \quad \text{subject to} \notag \\
& \qquad \sum\limits_{f \in \mc{F} } x_{i,s}(f) P_{i}(f) \le P_{U,i,w}^{(\ell)}, \quad  \forall i \in \mathcal{U} \tag{\ref{PB:WA}.a}\label{PB:PUweak} \nonumber\\
& \qquad  \sum\limits_{f \in\mc{F}} \sum\limits_{i \in \mathcal{D}} x_{i,s}(f){P}_{i,f} \le P_{D,w}^{(\ell)}, \tag{\ref{PB:WA}.b}\label{PB:PDweak} \\
&  \qquad \Gamma_{k',k}(f) \geq 0,\, \begin{aligned} \forall k',k &\in \mc{D} \\ \forall f &\in \mc{F} \end{aligned} \,\Big|\, \begin{aligned} &x_{k',w}(f) = x_{k,s}(f) = 1,\\
        &P_{k}^{(\ell)}(f)>0, \end{aligned}  \tag{\ref{PB:WA}.c}\label{PB:WSWWweak}
\end{align}

%In this case, we have relaxed the binding condition $P_{k'}(f) > 0$ of constraint~\eqref{PB:WSWW}, for the weak user $k' \in \mc{D}$. The same consideration about sub-optimality and feasibility of the strong allocation problem can be made.

It can be easily seen that problem~\eqref{PB:WA} has the same formulation of problem~\eqref{PB:SA}, with different optimization variables. Hence, the weak allocation problem can be efficiently solved in the \emph{Lagrangian dual domain}, as the previous strong allocation problem, obtaining
\begin{equation} 
\label{eq:LagDPweak}
\min_{\bm{\mu}\succeq 0} g_w(\bm{\mu}),
\end{equation}
where $\bm{\mu} = [\mu_0, \mu_1, \dots \mu_M]^T \in \mathbb{R}_{+}^{(U+1)}$ is the vector of the Lagrangian dual variables  associated to the constraints~\eqref{PB:PDweak} for $\mu_0$, and~\eqref{PB:PUweak} for $\mu_i$, $i \in 1, \dots M$.
The Lagrangian dual function $g_w(\bm{\mu})$ is
\begin{equation} 
\label{eq:LagWeak}
\begin{aligned}
&g_w(\bm{\mu})  = \hspace{-3pt} \max_{\substack{\mb{x} \in \mc{X}_w^{(\ell)}\\ \mathbf{P} \in \mc{P}_{w}^{(\ell)}}} \hspace{-3pt} U\left(\mb{x}_s^{(\ell)},\mb{P}_s^{(\ell)};\mb{x},\mb{P}\right) - K^{(\ell)} || \mb{P} - \mb{P}_w^{(\ell -1)}||^2  \\ &+ \sum_{j \in \mc{U}} \mu_j  \big( P_{U,i,w}^{(\ell)} - \hspace{-3pt} \sum_{f\in\mc{F}} P_{j}(f) \big)+ \mu_{0} \big(P_{D,w}^{(\ell)} \hspace{-2pt}- \hspace{-2pt} \sum_{f\in\mc{F}} \sum_{k \in \mc{D}} P_{k}(f) \big),
\end{aligned}
\end{equation}
where
%\begin{equation} \label{eq:feasibleStrong}
%\begin{aligned}
%\mc{P}_s^{(\ell)} = \big\{\mb{P} \tc 0 \le &P_j(f) \le P_U^{(\ell)}, \forall j \in \mc{U},\\   0 \le &P_{k}(f) \le P_D^{(\ell)}, \forall k \in \mc{D}\big\} \bigcap \\
%\big\{P_k'(f)\Gamma_{k',k} \ge 0&, \, \forall k',k \tc x_{k',w}(f) = x_{k,s}(f) = 1,\, f \in \mc{F}\big\}.
%\end{aligned}
%\end{equation}
\begin{equation} \label{eq:feasibleWeak}
\begin{aligned}
%\mc{P}_{w}^{(\ell)} = \Big\{&\mb{P} \tc \mb{P} \succeq 0,\, P_{k}^{(\ell-1)}(f) \Gamma_{k',k}(f) \geq 0, \\ &\forall k',k \in \mc{D} \tc x_{k',w}(f) = x_{k,s}(f) = 1,\,  \forall f \Big\}.
\mc{P}_{w}^{(\ell)} = \Big\{&\mb{P} \tc \mb{P} \succeq 0,\, P_{k}^{(\ell-1)}(f) \Gamma_{k',k}(f) \geq 0, \,\forall k',k \in \mc{D} \tc x_{k',w}(f) = x_{k,s}(f) = 1,\,  \forall f \Big\}.
\end{aligned}
\end{equation}
The Lagrangian auxiliary function which accounts for the contribution of the weak uplink-downlink allocation ($j,k$) on subcarrier $f$ is
\begin{equation}
\label{eq:auxLagWeak}
\begin{aligned}
L_{j,k,w}&(\mb{P}(f),\bm{\mu}) = \alpha_j R_j(f) + \alpha_k R_k(f) + U_s(\mb{x}_s^{(\ell)}, \mb{P}(f))\\
 &- K_\ell ||[P_j(f), P_k(f)] - [P_j^{(\ell-1)} (f),P_k^{(\ell-1)}(f)]||^2
- \mu_j P_j(f) - \mu_0 P_{k}(f).
\end{aligned}
\end{equation}
Hence, neglecting  the terms not involved in the optimization, the dual Lagrangian function~\eqref{eq:LagWeak} is equivalent to
\begin{equation}
\label{eq:lagCompWeak}
\begin{aligned}
g_w(\bm{\mu})\hspace{-2pt} = \hspace{-4pt}\sum_{f\in\mc{F}} \hspace{-4pt} \max_{\substack{\mb{x} \in \mc{X}_w^{(\ell)} \\ \mathbf{P}(f) \in \mc{P}_w^{(\ell)}}} \hspace{-3pt} \sum_{j \in \mc{U}} \sum_{k \in \mc{D}} x_{j,s}(f) x_{k,s}(f) L_{j,k,w}(\mb{P}(f),\bm{\mu}).
\end{aligned}
\end{equation}
Also in this case, the resolution of the maximization problem~\eqref{eq:lagCompWeak} is obtained computing first the power coefficients and then the allocation variables.

\FS{The power coefficients are evaluated by maximizing the auxiliary Lagrangian function~\eqref{eq:auxLagWeak}, keeping fixed a single weak couple $j \in \mc{U}$, $k \in \mc{D}$, and the previously selected strong users $\mb{x}_s^{(\ell)}$ and $\mb{P}^{(\ell)}_s$. Hence, for \emph{each possible couple of weak users} $j \in \mc{U}$ and $k \in \mc{D}$, we solve}
\begin{equation} \label{eq:maxLweak}
\begin{aligned}
\left[P^{*}_{j}(f),P^{*}_{k}(f)\right]=\argmax_{\mb{P}(f) \in \mc{P}_w^{(\ell)}} & \, L_{j,k,w}(\mb{P}(f),\bm{\mu}), \quad 
\begin{aligned}
j &\in \mc{U}, \\
k &\in \mc{D},
\end{aligned}
\end{aligned}
\end{equation}
Problem~\eqref{eq:maxLweak} is still a DC function maximized on a convex set, and its solution is described in Appendix~\ref{sec:powerAllocation}.

Once the  power coefficients are computed for every possible couple of users, for each subcarrier, the optimal allocation variables are obtained by selecting the couple which maximizes the Lagrangian function, i.e.,
\begin{equation} \label{eq:alloVarWeak}
x_{j(f),w}(f) = 1, x_{k(f),w}(f) = 1, \quad \forall f \in\mc{F},
\end{equation}
where $j(f)\in \mc{U}$ and $k(f) \in \mc{D}$ and it is 
\begin{equation}
[j(f), k(f)]= \arg\max_{j,k} L_{j,k,w}(\mb{P}^*(f),\bm{\mu}), \,\, f \in \mc{F}.
\end{equation}
\FS{All the allocation variable and the power coefficient of users not allocated %as weak 
are set to 0, i.e. %$\forall f \in \mc{F}$
\begin{equation}
\begin{aligned}
x_{j,w}(f) = 0, \quad& P_j(f) = 0,  \\
x_{k,w}(f) = 0, \quad& P_k(f) = 0,
\end{aligned} \quad
\begin{aligned}
\forall j \neq j(f) \in \mc{U}, \\
\forall k \neq k(f) \in \mc{D}.
\end{aligned}
\quad \forall f\in \mc{F}
\end{equation}}

The \emph{cutting plane ellipsoid} method~\cite{Yu2006} is employed to update $\bm{\mu}$ until convergence, as described in Appendix~\ref{sec:dual}.

The Weak user Allocation (WA) algorithm is summarized in Algorithm~\ref{alg:WA}. The main complexity burden is represented by computing  $F M N$ metrics at each iteration on $\bm{\mu}^{(t)}$.
At  convergence, the algorithm yields to local optimum variables, collected in vectors $\mb{x}_w^{(\ell)}$ and $\mb{P}_{w}^{(\ell)}$.

%%%%%%  ALGORITHM  %%%%%%
\begin{algorithm}
\footnotesize
\caption{Weak user Allocation (WA)}
\label{alg:WA}
\textbf{Input:}
$\mathbf{P} \leftarrow [\mb{P}_s^{(\ell)}, \mb{P}_w^{(\ell-1)}]$, $t \leftarrow 1$, $\Delta \leftarrow 1$; \textbf{Output:} $\mb{P}_w^{(\ell)}$, $\mb{x}_w^{(\ell)}$\;
\While{$\Delta > 0$}
{
	\For {$f \in \mc{F}$}
	{
		\For {$(j,k) \in (\mathcal{U},\mathcal{D})$}
		{			
			Compute $P_{j}(f)$ and $P_{k}(f)$ by solving~\eqref{eq:maxLweak};
		}
		Compute $x_{i,s}(f)$ as in~\eqref{eq:alloVarWeak};\\		
	}
	Update $\bm{\mu}^{(t)}$ as described in Appedix~\ref{sec:dual}\;
	$\Delta \leftarrow \left\Vert \bm{\mu}^{(t)} - \bm{\mu}^{(t-1)}\right\Vert$\;
	$t \leftarrow t + 1$\;
}
\end{algorithm}
\normalsize

%==================================================
\section{A low-complexity algorithm}
%==================================================
\label{sec:low-complexity}

\MM{Even if the BCD algorithm always converges to a local optimum, its  complexity depends on the number of iterations needed for  convergence, which is unknown \emph{a priori}. To address this issue, we propose here  a simplified  algorithm, which exploits the BCD structure but has a lower complexity and comparable performance.
The general idea is based on the observation  that, because of the specific structure of our system, which is based on interference cancellation,  the most important allocation decisions of the BCD algorithm are taken during the first iteration, i.e. for $\ell=1$, while in the remaining iterations the algorithm mainly switches between similar allocations and slowly  updates the power coefficients. 
In fact, since interference is cancelled  at the  receivers of  both uplink and downlink strong users, to a first approximation we can consider the allocation of the strong users almost decoupled from the allocation of the weak users. Being the channel allocations of  the strong users almost indipendent from the  weak users, it is the power allocation that takes a certain number of iterations before achieving convergence. Accordingly, we propose a low-complexity \FS{(LC)} algorithm divided in three steps: first, the strong users are allocated by means of the SA algorithm, second the WA algorithm is performed to allocate the weak users and, finally, a \emph{power redistribution algorithm} (\emph{PRA}) is employed in order to distribute the power available on the allocated users.}
%\FS{It is worth noting that the regularization term $K^{(1)}$ employed in the SA and WA procedures is set to 0. In fact, for a single iteration of the BCD process the regularization is irrelevant.}

\subsection{Power Redistribution Algorithm (PRA)}
\label{sec:REPOW}

\MM{The PRA is designed  to  optimally distribute the power coefficients once the channel allocations are fixed; hence, allocation variables $\mb{x} = [\mb{x}_s, \mb{x}_w]$ obtained by the resolution of problems~\eqref{PB:SA} and~\eqref{PB:WA} are fixed and not part of the optimization problem.}

Let us denote as $\mb{P}_U$ the vector collecting the power coefficients of the users allocated in the uplink, indifferently if weak or strong, i.e., $\mb{P}_U = \{ P_{i}(f), \, \forall i \in \mc{U} \tc x_{i,s}(f) = 1 \text{ or } x_{i,w}(f) = 1, \, \forall f\in \mc{F}\}$. Similarly, we denote as $\mb{P}_D$ the vector collecting the power coefficients of the users allocated in the downlink, i.e., $\mb{P}_D = \{ P_{i}(f), \, \forall i \in \mc{D} \tc x_{i,s}(f) = 1 \text{ or } x_{i,w}(f) = 1, \, \forall f \in \mc{F}\}$.
Under these assumptions, the power allocation problem can be rewritten as:
\begin{equation}
\begin{aligned}
\max_{\mb{P}_U,\mb{P}_D}& U(\mb{x},\mb{P}_U,\mb{P}_D) \\
\text{subject to }& \eqref{PB:PU},\,\eqref{PB:PD},\,\eqref{PB:WSWW}.
\end{aligned}
\end{equation}
Once again, we can use Lagrangian dual decomposition to relax the power budget constraints, obtaining the equivalent problem
\begin{equation}
\label{eq:dualsequential}
g(\bm{\mu}) = \sum_{f=1}^F  \max_{\mb{P}_U(f),\mb{P}_D(f)} L(\mb{P}_U(f),\mb{P}_D(f)),
\end{equation}
subject to NOMA constraint~\eqref{PB:WSWW} and with
\begin{equation} \label{eq:maxLf}
\begin{aligned}
L(\mb{P}_U(f),\mb{P}_D(f)) &= U_s(\mb{x}_s,\mb{P}_U(f),\mb{P}_D(f))
+ U_w(\mb{x}_w,\mb{P}_U(f),\mb{P}_D(f))\\
&-  \sum_{j \in \mc{U}(f)} \mu_j P_{j}(f) -  \sum_{k \in \mc{D}(f)} \mu_0 P_{k}(f),
\end{aligned}
\end{equation}
where $\mb{P}_U(f)$ and $\mb{P}_D(f)$ refer to the power coefficients of the users allocated on subcarriers $f$, $\mu_0$ is the Lagrangian variable associated to constraint~\eqref{PB:PD}, and $\mu_j$, $j = 1, \dots, M$, is the Lagrangian variable associated to constraint~\eqref{PB:PU}.

Problem~\eqref{eq:dualsequential} is now separable for each subcarrier. Furthermore, it is sum of $2A = 4$ difference of convex function due to the $2A$ users allocated on each subcarrier; hence, it is a DC function, where the expression of the concave $L_\text{cav}$ and convex $L_\text{vex}$ part are reported in~\eqref{eq:dc}.
In order circumvent the limitation of the NOMA constraint~\eqref{PB:WSWW}, we formalize a sequential version of the concave-convex procedure~\cite{Yuille}, firstly for downlink and then for uplink direction. In particular, the feasible set of the uplink problem consider the constraint only if the downlink power coefficient evaluated at that iteration are both non-zero.
At the $(t+1)$-th iteration, the power allocation problem can be expressed by:
\begin{equation} \label{eq:seq}
\begin{aligned}
%\mb{P}_D^{(t+1)}(f) &= \argmax_{\mb{P}_D(f) \in \mc{P}_D} L_\text{cav}(\mb{P}_U^{(t)}(f),\mb{P}_D(f)) \\& \qquad \quad +\mb{P}_D^T \nabla_{\mb{P}_D(f)} L_\text{vex}(\mb{P}_U^{(t)}(f),\mb{P}_D^{(t)}(f)),  \\
%\mb{P}_U^{(t+1)}(f) &= \argmax_{\mb{P}_U(f) \in \mc{P}_U} L_\text{cav}(\mb{P}_U(f),\mb{P}_D^{(t+1)}(f)) \\& \qquad \quad + \mb{P}_U^T \nabla_{\mb{P}_U(f)} L_\text{vex}(\mb{P}_U^{(t)}(f),\mb{P}_D^{(t+1)}(f)),
\mb{P}_D^{(t+1)}(f) &= \argmax_{\mb{P}_D(f) \in \mc{P}_D} L_\text{cav}(\mb{P}_U^{(t)}(f),\mb{P}_D(f)) +\mb{P}_D^T \nabla_{\mb{P}_D(f)} L_\text{vex}(\mb{P}_U^{(t)}(f),\mb{P}_D^{(t)}(f)),  \\
\mb{P}_U^{(t+1)}(f) &= \argmax_{\mb{P}_U(f) \in \mc{P}_U} L_\text{cav}(\mb{P}_U(f),\mb{P}_D^{(t+1)}(f)) + \mb{P}_U^T \nabla_{\mb{P}_U(f)} L_\text{vex}(\mb{P}_U^{(t)}(f),\mb{P}_D^{(t+1)}(f)),
\end{aligned}
\end{equation}
where $\mc{P}_D = \{0 \preceq \mb{P}_D \preceq P_D\}$, $\mc{P}_U = \{0 \preceq \mb{P}_U \preceq P_U, P_{k}^{(t+1)}(f)P_{k'}^{(t+1)}(f)\Gamma_{k',k}(f) \geq 0,\, \forall k,k'\in\mc{D}(f)\}$. Note that $P_U$ and $P_D$ are inserted as upper bounds to prevent the algorithm to diverge; the NOMA constraint $\Gamma_{k',k}(f)$ is considered only if both the downlink power coefficients are non-zero during the same iteration.
The convergence of procedure~\eqref{eq:seq} is given by the following theorem.
\begin{theorem}
\label{theo:cccpSEQ}
The sequential procedure given in~\eqref{eq:seq}  generates a sequence  of power values such that at iteration $t+1$ it is  
\begin{equation}
L(\mb{P}_U^{(t+1)}(f),\mb{P}_D^{(t+1)}(f)) %\geq L(\mb{P}_U^{(l)}(f),\mb{P}_D^{(l+1)}(f)) 
\geq L(\mb{P}_U^{(t)}(f),\mb{P}_D^{(t)}(f)), \nonumber
\end{equation}
 which finally converges on local optimum $\mb{P}^{*}(f)$.
\begin{proof}
Please refer to Appendix~\ref{sec:dim2}
\end{proof}
\end{theorem}
Finally, the optimal dual variables are updated as described in Appendix~\ref{sec:dual}.
The power redistribution algorithm proposed is summarized in Algorithm~\ref{alg:seq}.

\begin{algorithm}
\footnotesize
\caption{Power Redistribution Algorithm (PRA)}\label{alg:seq}
\textbf{Input:}
$\mathbf{x} \leftarrow$ solution of~\eqref{PB:SA} and~\eqref{PB:WA}, $t \leftarrow 1$, $\Delta \leftarrow 1$; \textbf{Output:} $\mb{P}$\;
\While{$\Delta > 0$}
{
	\For {$f \in \mc{F}$}
	{
	Compute $\mb{P}_{U}(f)$ and $\mb{P}_{D}(f)$ by sequential programming~\eqref{eq:seq};		
	}
	Update $\bm{\mu}^{(t)}$ as described in Appedix~\ref{sec:dual}\;
	$\Delta \leftarrow \left\Vert \bm{\mu}^{(t)} - \bm{\mu}^{(t-1)}\right\Vert$\;
	$t \leftarrow t + 1$\;
}
\end{algorithm}
\normalsize

%===============================================
\section{Numerical Results}
%===============================================
\label{sec:results}
This Section will provide an evaluation of the performance of the proposed schemes.
\FS{To simplify the notation, we denote as BCD the overall iterative procedure defined in~\eqref{eq:updateRule}, and as LC the low-complexity approach given by using SA, WA and PRA in sequence, as described in Section \ref{sec:low-complexity}.}
As a term of comparison, we also show the results for the scheme proposed in~\cite{Abrardo2019}, referred as WMMSE, and the two optimization schemes proposed in~\cite{Sun2017}, labeled as REF and SCA, respectively. It is worth observing that all algorithms, with the sole exception of  REF, make use of sub-optimal iterative approaches that, at each iteration, solve an intermediate optimization problem to ultimately obtain a local optimum of the weighted sum-rate. The optimal reference algorithm, REF, is still an iterative approach, but it is based on  monotonic optimization theory, as proposed in~\cite{Zhang2010}, which guarantees the convergence to the global optimum. For comparison purposes we also plot the results obtained for the OMA-FD case, obtained using the approach proposed in~\cite{Xiao17}, which is shown to be quasi-optimal.

To get an insight of the complexity of the various schemes, we give in Table~\ref{TableComplexity} the computational complexity of the algorithms in terms of number of elementary operations of the dominant term per iteration, following the approach in~\cite{Moretti13}. \FS{In the same table we also present the number of iterations needed for convergence.} 
It is worth noting that, even if REF and SCA have the same complexity per iteration, the latter requires a much lower number of iterations~\cite{Sun2017}. 
For the algorithms we propose, the separation of the allocation of strong and weak users in two distinct steps noticeably reduces \FS{the complexity per iteration} with respect to REF and SCA. \FS{However, the BCD algorithm needs a certain number of SA-WA iteration to converge, and, in turn, each step needs a certain number of iterations of the ellipsoid method~\cite{Yu2006}. This lead to the largest number of iterations of all compared schemes.}
\FS{On the other hand, the LC algorithm shows a negligible complexity overhead per iteration respect to the BCD, but it only needs the iterations for the ellipsoid method.} \FS{For the proposed algorithms, we presented in the table the maximum number of iteration needed for the convergence with the maximum number of users considered, i.e. $M = 50$.}
\FS{In Fig.~\ref{fig:complexity}, we show the overall computation complexities as a function of $M$, in the case of the same number of uplink and downlink users, i.e., $M = N$.}
\FS{It can be seen that the BCD algorithm is computationally cheaper respect to both REF and SCA when the number of users is more than 10. In every case, the LC approach has always a lower complexity.}
Finally, the WMMSE approach presented in~\cite{Abrardo2019} has the lowest complexity among all considered solutions.

\begin{figure}[ht]
\centering
\begin{floatrow}
\capbtabbox{%
\footnotesize
\begin{tabular}{cp{65pt}p{40pt}} 
\toprule
Algorithm & Complexity per~iteration & Number of iterations\\
\midrule
REF & $O( F M^2 N^2 )$& $\approx500$\\
SCA & $O( F M^2 N^2 )$& $\approx50$\\
BCD & $O( 2 F M N )$& $\leq3600$\\
LC & $O( F (2 M N +1))$ & $\leq120$\\
WMMSE & $O( F (M + N))$& $\approx200$\\
\bottomrule
\end{tabular}
\normalsize}
{\caption{Complexity comparison of the algorithms.} \label{TableComplexity}}
\ffigbox{%
%\begin{subfigure}[c]{0.45\textwidth}
\includegraphics[width=7 cm]{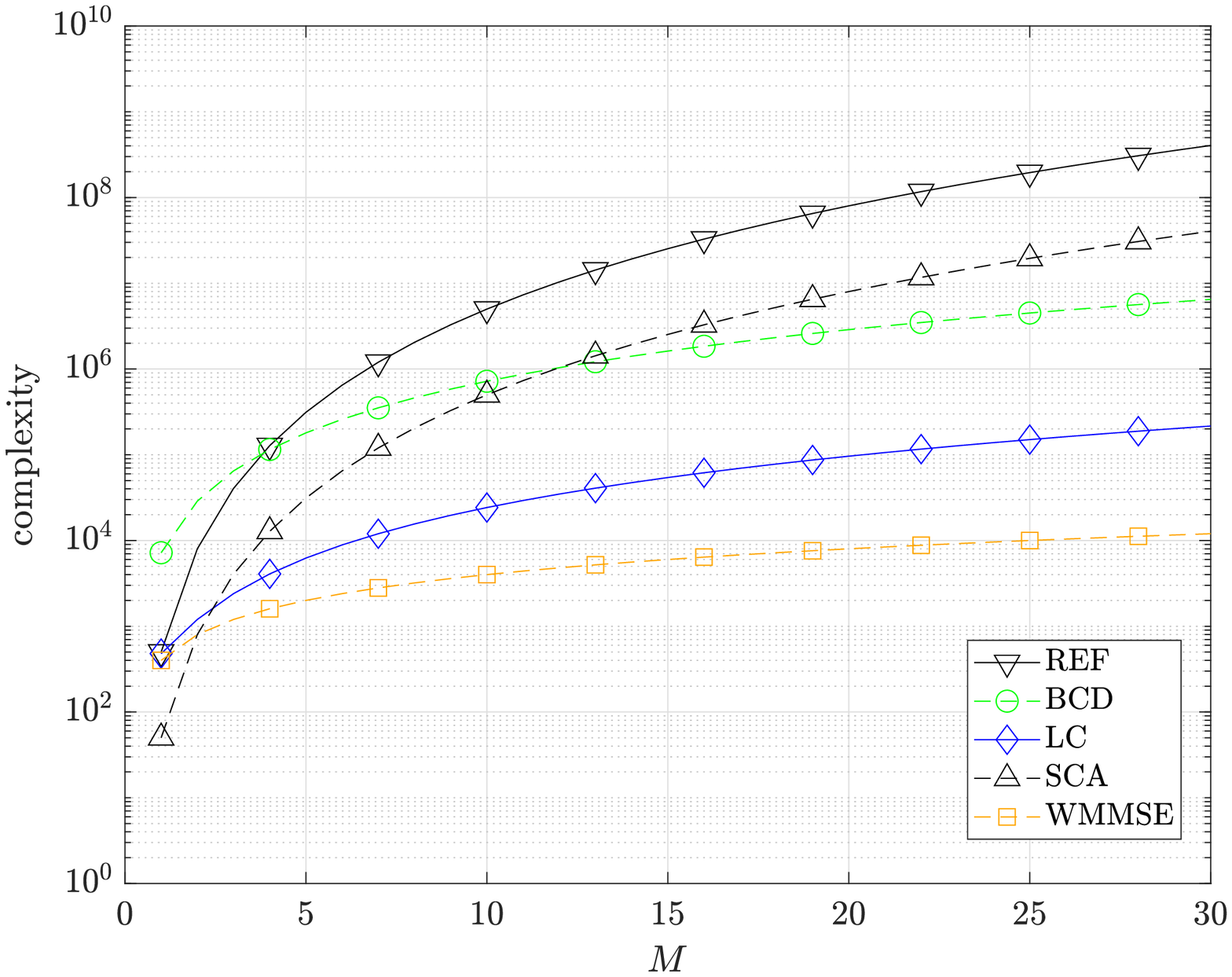}}{
\caption{Overall complexity as a function of the number of users $M$, when $M = N$.}
\label{fig:complexity}}
\end{floatrow}
\end{figure}

\FS{To examine the performance of the allocation schemes, we adopt the Monte Carlo method, averaging the results for 1000 of simulation instances.} The simulated scenario is a single \FS{small} cell with radius $R=100$ m and $F = 6$ subcarriers with the same number of uplink and downlink users, i.e., $M = N$. \MM{The  number of channels and users  is set  to  test the network in overload conditions.} %\FS{The center frequency is 2 GHz and the bandwidth of each subcarrier is 180 kHz.} 
\FS{The noise power is set to -121 dBm for all receiver.}
At each instance of the Monte Carlo simulations, the positions of the mobile users are randomly generated in the cell with a minimum distance from the serving BS of $30$ m. 
\FS{To simulate an outdoor scenario in an urban area, the propagation channel is generated considering the presence of log-normally distributed shadowing with standard deviation $\sigma_{SH}= 8$ dB; the path loss attenuation is exponentially proportional to the distance between the transmitters and receivers with the path loss exponent $\delta = 4$.} For each subcarrier, we assume uncorrelated fading  with channel coefficients generated from the complex Gaussian distribution $\mathcal{CN}(0,1)$. The SI cancelation factor $\zeta(f)$, $f\in\mc{F}$, at the BSs is set to a constant value of 110 dB, that is a reasonable value for the considered scenario~\cite{Goyal17}. To achieve a a certain degree of fairness of the allocation, \FS{especially for the edge users,} the weights of the utility function are proportional to the distance  $d_i$  between user $i$ and the BS, i.e.,  $\alpha_i = | \frac{d_i}{\max_{i \in \mc{U},\mc{D}} d_i}|^2$.

\MM{Fig.~\ref{fig:BCDvsLC} compares the performance of the BCD and the LC approaches by showing the utility $U$ as a function of the steps of the BCD algorithm, for $M = N = 6$, $P_U = 14$ dBm, $P_D = 24$ dBm.} \MM{It is worth noting that the odd steps represent the strong users'  allocations while the even steps represents the weak users' allocations. The solution obtained by the LC and REF algorithms are also plotted. Similar plots are obtained with different set of parameters. Firstly, we can note the impact of the regularization term on the procedure. When the regularization is not employed, the monotone behaviour of the sum rate is not assured and Theorem~\ref{theo:convBCD} does not hold any more. Otherwise, a correct value of $K^{(\ell)}$ leads to the proper convergence. Now, let us focus on the performance of the LC algorithm, which obtains the same performances of the BCD at the convergence. For this motivation, the following plot will present the LC performances only.}

%\begin{figure}[tbh]
%\centering
%\ffigbox{
%\includegraphics[width=7 cm]{convergence_all.eps}}{
%\caption{$U$ as a function of the steps of the BCD algorithm compared with the result obtained by LC, for $M = N = 6$, $F = 6$, $P_U = 14$ dBm, $P_D = 24$ dBm.}}
%\label{fig:BCDvsLC}
%\end{figure}

Figs.~\ref{fig:3userPU} and~\ref{fig:3userPD} show $U$ as a function of the available maximum transmit powers in uplink and downlink, respectively, for $M = N = 6$. In particular, in Fig.~\ref{fig:3userPU} the maximum transmit power $P_D$ in the downlink is set to 20 dBm, whereas in Fig.~\ref{fig:3userPD} the maximum transmit power $P_U$ in the uplink is set equal to $14$ dBm. In all cases the proposed \FS{LC scheme} clearly outperform the OMA-FD scheme, SCA and WMMSE. %,  while FPWA has similar performances to WMMSE. 
The performance gap between  \FS{LC} and REF schemes is negligible, showing that it is  possible to obtain performance close to the optimum with a much lower implementation complexity. It is of particular importance the comparison with SCA, which is another 'low-complexity' algorithm designed to address the NOMA-FD problem. While in some scenarios REF and SCA are very close~\cite{Sun2017}, in our setting SCA performs  worse than both REF and LC. This probably happens because SCA performance are very sensitive to the choice of input parameters, which need to be optimized with great care for each different scenario.   

\begin{figure}[bth]
\centering
\begin{subfigure}[t]{0.45\textwidth}
\includegraphics[width=7 cm]{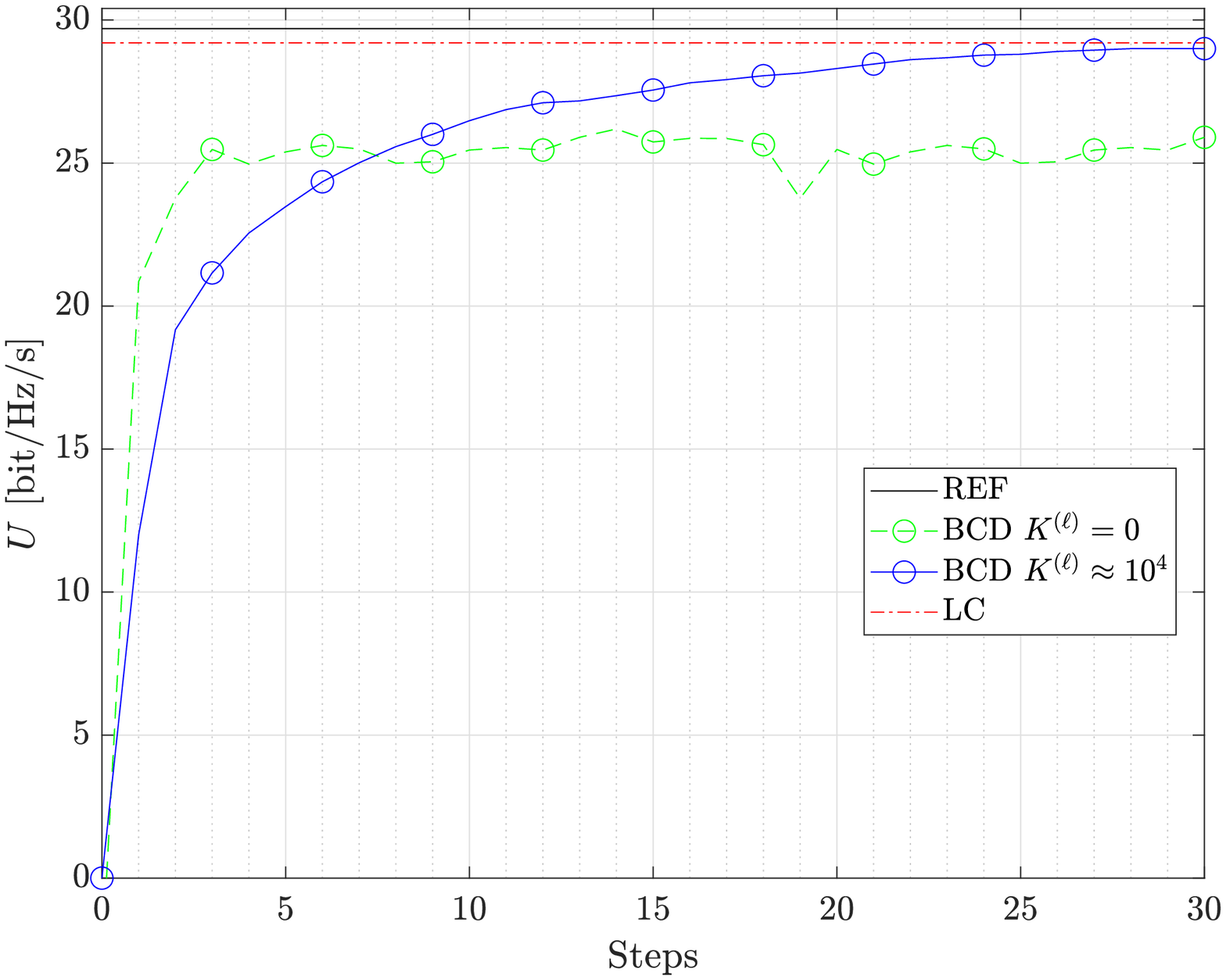}
\caption{$U$ vs the BCD steps compared with the result obtained by LC and REF, for $P_U = 14$ dBm, $P_D = 24$ dBm.}
\label{fig:BCDvsLC}
\end{subfigure}
~~
\begin{subfigure}[t]{0.45\textwidth}
\includegraphics[width=7 cm]{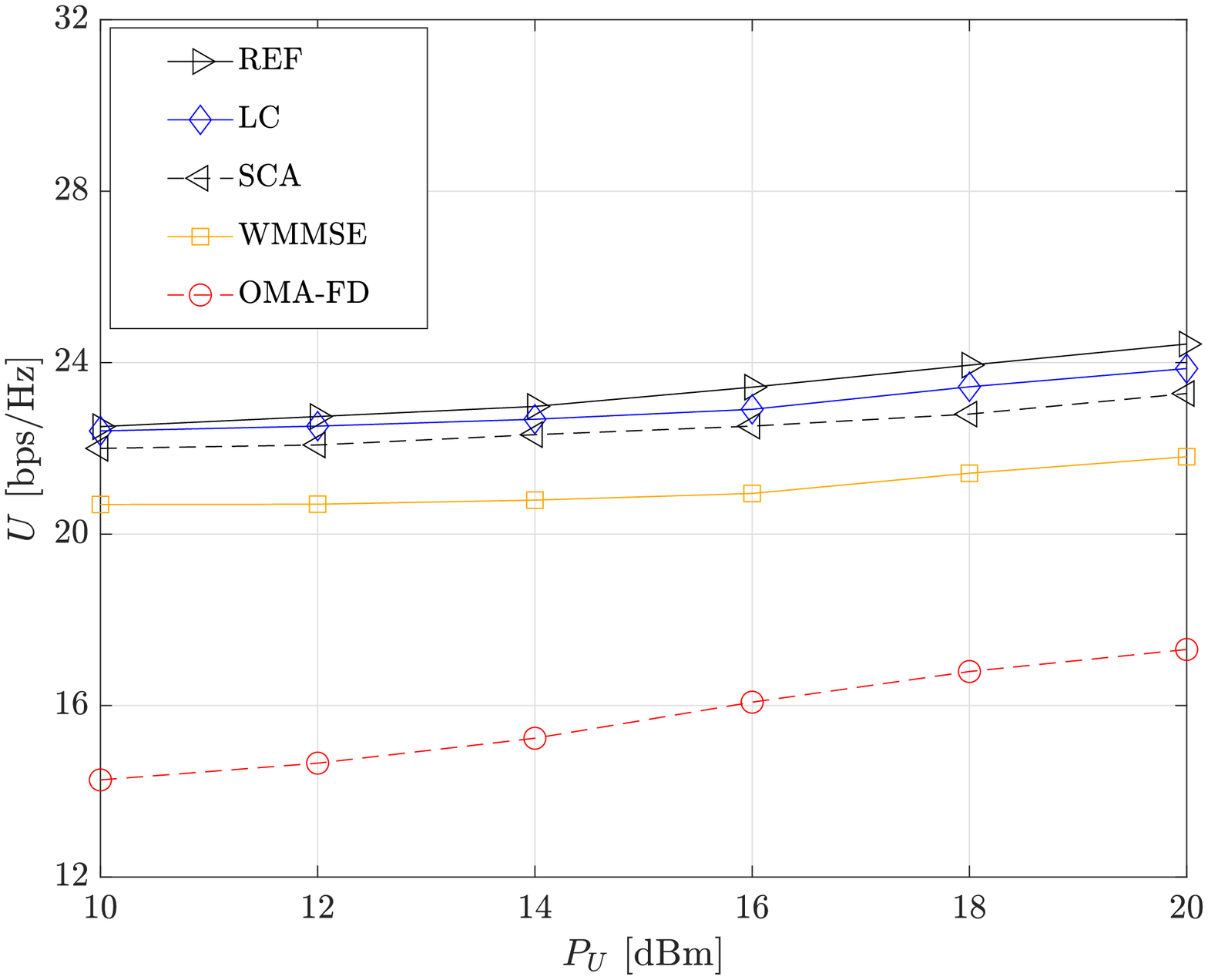}
\caption{$U$ as a function of the maximum transmitting powers in uplink, $P_D$ = 20 dBm.}
\label{fig:3userPU}
\end{subfigure}
\\
\begin{subfigure}[t]{0.45\textwidth}
\includegraphics[width=7 cm]{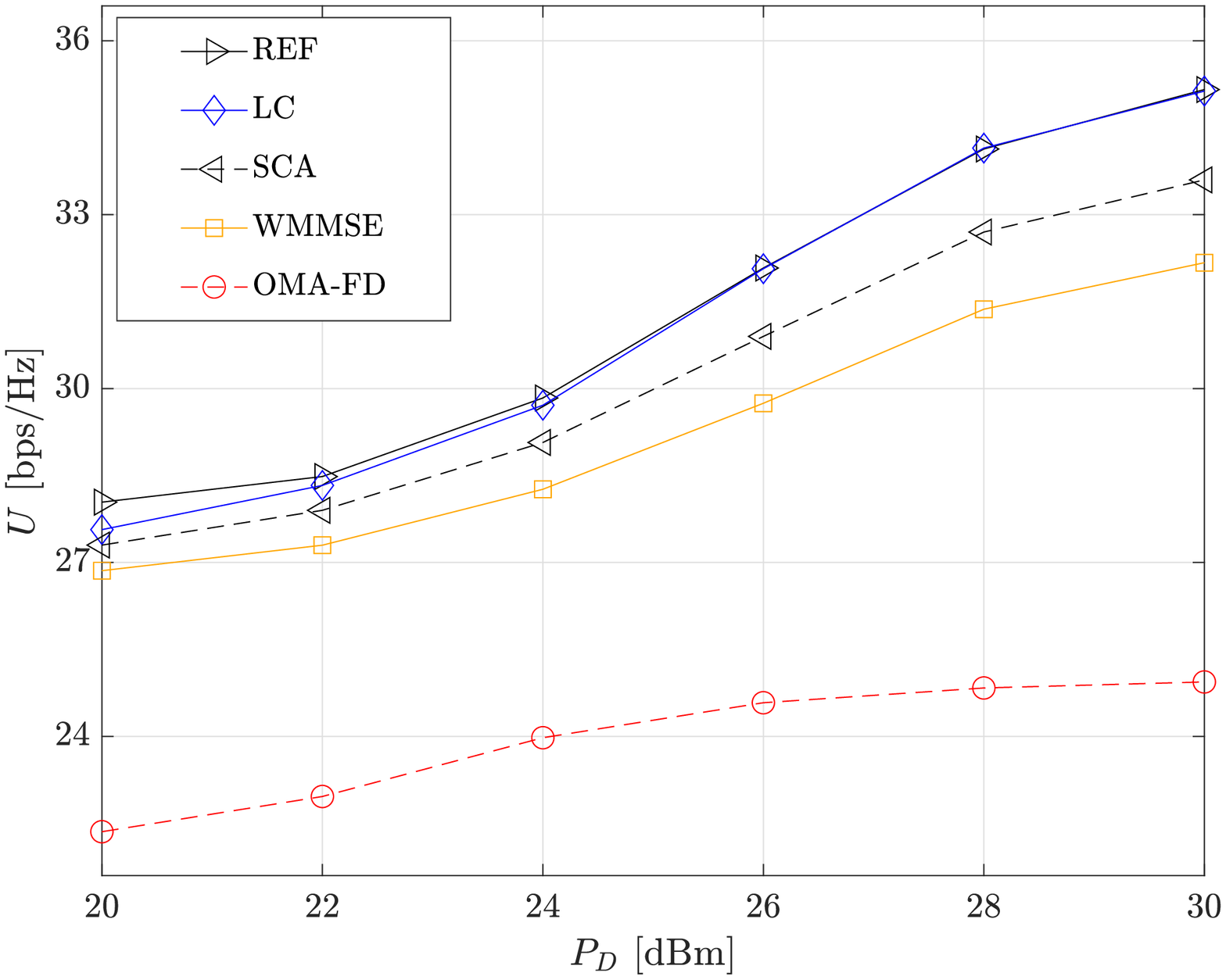}
\caption{$U$ as a function of the maximum transmitting power in downlink, $P_U$ = 14 dBm.}
\label{fig:3userPD}
\end{subfigure}
~~
\begin{subfigure}[t]{0.45\textwidth}
\centering
\includegraphics[width=7 cm]{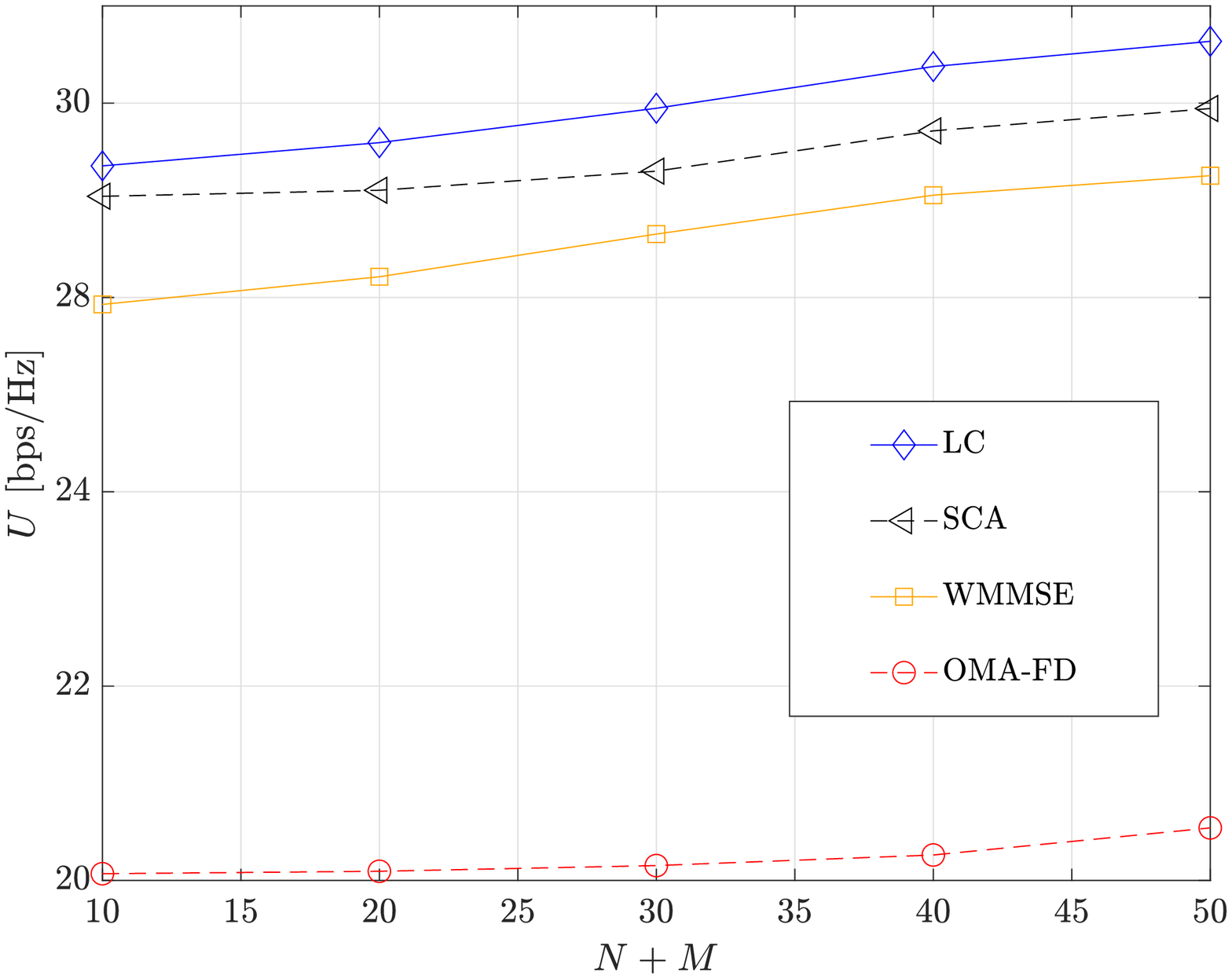}
\caption{$U$ as a function of the total number of user in the cell $N+M$, $N = M$, $P_U = 14$ dBm, $P_D = 20$ dBm.}
\label{fig:varN}
\end{subfigure}
\caption{Performance evaluation for $F=6$. For a) b) and c) the number of users is $M = N = 6$. }
\end{figure}

Due to the simulations'  computational load, the results obtained for more crowded scenarios are reported in the next figure for \FS{LC}, SCA, and WMMSE schemes only. 
Fig.~\ref{fig:varN} shows $U$ as a function of the number of users in the cell, for $P_U = 14$ dBm and $P_D = 20$ dBm. As expected, all NOMA schemes outperform  OMA for any number of users. Moreover, once again LC attains the best results,  outperforming \FS{both} SCA \FS{and WMMSE}.  %Also in this case, WMMSE and FPWA have similar performances.

\section{Conclusions} \label{sec:conclusions}

In this paper, the power and channel allocation problem for multicarrier non-orthogonal multiple access (NOMA) full duplex (FD) systems has been investigated. Following the block coordinated descent (BCD) approach, we have proposed two algorithm based on the decomposition of the original allocation problem in lower-complexity sub-problems solvable in the Lagrangian dual domain with a great reduction of the computational load. The proposed approaches allow to achieve performances that are reasonably close to the optimum and outperforms more complex algorithms addressing the same optimization problem.

% ================================================
\begin{appendices}
\section{Solution of power allocation problem}
\label{sec:powerAllocation}
\renewcommand{\theequation}{\thesection.\arabic{equation}}
\setcounter{equation}{0}

Considering an iteration $\ell$ of the BCD algorithm, power allocation problems~\eqref{eq:maxLstrong},~\eqref{eq:maxLweak} can be expressed by the following general formulation:
\begin{equation}
\label{eq:maxL}
[P_{j}^*(f), P_{k}^*(f)] = \argmax_{\mb{P}(f) \in \mc{P}_{m}^{(\ell)}} L_{j,k,m}(\mathbf{P}(f), \bm{\mu}), \, \,
\begin{aligned}
j &\in \mc{U}, \\
k &\in \mc{D},
\end{aligned}
\end{equation}
with $m \in \{s,w\}$ indicating if the problem regards the SA or the WA and, accordingly, the feasible set considered.
Let us firstly focus to provide a convex and bounded formalization for $\mc{P}_{s}^{(\ell)}$ and $\mc{P}_{w}^{(\ell)}$. For convenience, we will focus on the strong feasible set only; however, the description can be easily adapted for the weak feasible set.
Without loss of generalization, feasible set~\eqref{eq:feasibleStrong} can be expressed by taking only into account the couple of users under examination $j \in \mc{U}$ and $k \in \mc{D}$ as $\mc{P}_{s}^{(\ell)} = \big\{0 \leq P_{j}(f) \leq P_U, 0 \leq P_{k}(f) \leq P_D, P_{k'}^{(\ell - 1)} (f)\Gamma_{k',k}(f) \ge 0 \tc x_{k,s}(f) = x_{k',w}(f)  = 1\big\}.$
%\begin{equation}
%\begin{aligned}
%\mc{P}_{s}^{(\ell)} &= \big\{0 \leq P_{j}(f) \leq P_U, 0 \leq P_{k}(f) \leq P_D, 	
%P_{k'}^{(\ell - 1)} (f)\Gamma_{k',k}(f) \ge 0 \tc x_{k,s}(f) = x_{k',w}(f)  = 1\big\}. \nonumber
%\end{aligned}
%\end{equation}
If $P_{k'}^{(\ell-1)}(f) = 0$, the NOMA constraint is always verified.
Otherwise, if $P_{k'}^{(\ell-1)}(f) > 0$, the feasible set can be empty depending on the channel gains involved in $\Gamma_{k',k}(f)$. 
Indexing as $j' \in \mc{U}(f)$ and $k'\in\mc{D}(f)$ the previously allocated weak users on $f$, the expression of~\eqref{eq:thetaDef} can be simplified in $\Gamma_{k',k}(f) = \theta_{j,f}^{(k',k)} P_{j}(f) + \theta_{j',f}^{(k',k)} P_{j'}(f) + \delta_f^{(k',k)}.$ 
%\begin{equation}
%\begin{aligned}
%\Gamma_{k',k}(f) = \theta_{j,f}^{(k',k)} P_{j}(f) + \theta_{j',f}^{(k',k)} P_{j'}(f) + \delta_f^{(k',k)}. \nonumber
%\end{aligned}
%\end{equation}
Accordingly, the feasible set $\mathcal{P}_{s}^{(\ell)}$ can be expressed by
\begin{equation}
\label{eq:feasible}
\begin{aligned}
&\mathcal{P}_s^{(\ell)}= \left\{ 0 \leq P_{k,f} \leq P_D \right\} \cap 
\begin{cases}
0 \leq P_{j}(f) \leq \left(P_U, \bar{P}_j(f) \right)^-, &\theta_{j,f}^{(k',k)} < 0, P_{k'}^{(\ell-1)}(f) > 0 \\
\left(\bar{P}_j(f)\right)^+  \leq P_{j}(f) \leq P_U, &\theta_{j,f}^{(k',k)} > 0, P_{k'}^{(\ell-1)}(f) > 0  \\
0 \leq P_{j}(f) \leq P_U, &P_{k'}^{(\ell-1)}(f) = 0, \\
\emptyset, & \theta_{j,f}^{(k',k)} = 0
\end{cases}
\end{aligned}
\end{equation}
where $(a,b)^- = \min(a,b)$ and $(a)^+ = \max(0,a)$ and 
%\begin{equation*}
$\bar{P}_j(f) = - \frac{\theta_{j',f}^{(k',k)} P_{j'}(f) + \delta_f^{(k',k)}}{\theta_{j,f}^{(k',k)}}$.
%\end{equation*}
If not empty, set~\eqref{eq:feasible} is a convex and bounded. The same procedure can be applied on $\mc{P}_w^{(\ell)}$.
A visualization of this feasible set is shown in Figure~\ref{fig:feas2}. %Note that, for the weak feasible set, we can simply switch  user $j$ with $j'$ and $k$ with $k'$.

\begin{figure}[htb]
\centering
\begin{subfigure}[c]{0.3\textwidth}
\begin{tikzpicture}[scale = 0.60]
% assi
\draw [->] (0,-0.3) -- (0,4.5);
\draw [->] (-0.3,0) -- (5,0);
%rette
\draw (-0.3,3)--(4.5,3);
\draw (3,-0.3) -- (3,4.1);
\draw (4,-0.3) -- (4,4.1);
%punti d'intersezione e poliedro
\draw [thick, fill, fill opacity = 0.1] plot [mark=*] coordinates {(0,0)(0,3)(3,3)(3,0)(0,0)};
%nodi
\node at (1.5,1.5) {$\mathcal{P}_s^{(\ell)}$};
%\node at (-0.3,-0.3) {$0$};
\node at (4.8,-0.45) {$P_{j}(f)$};
\node at (-0.3,4.8) {$P_{k}(f)$};
\node at (5,3)  {$P_D$};
\node at (2.7,4.4){$\bar{P}_j(f)$};
\node at (4,4.5){$P_U$};
\end{tikzpicture}
\caption{$\theta_{j,f}^{(k',k)} > 0, P_{k'}^{(\ell-1)}(f) > 0$}
\end{subfigure}
~
\begin{subfigure}[c]{0.3\textwidth}
\begin{tikzpicture}[scale = 0.60]
% assi
\draw [->] (0,-0.3) -- (0,4.5);
\draw [->] (-0.3,0) -- (5,0);
%rette
\draw (-0.3,3)--(4.5,3);
\draw (1,-0.3) -- (1,4.1);
\draw (4,-0.3) -- (4,4.1);
%punti d'intersezione e poliedro
\draw [thick, fill, fill opacity = 0.1] plot [mark=*] coordinates {(1,0)(1,3)(4,3)(4,0)(1,0)};
%nodi
\node at (2.5,1.5) {$\mc{P}_s^{(\ell)}$};
%\node at (-0.3,-0.3) {$O$};
\node at (5.3,-0.45) {$P_{j}(f)$};
\node at (-0.3,4.8) {$P_{k}(f)$};
\node at (5,3)  {$P_D$};
\node at (1.5,4.5){$\bar{P}_j(f)$};
\node at (4,4.5){$P_U$};
\end{tikzpicture}
\caption{$\theta_{j,f}^{(k',k)} < 0, P_{k'}^{(\ell-1)}(f) > 0$}
\end{subfigure}
~
\begin{subfigure}[c]{0.3\textwidth}
\begin{tikzpicture}[scale = 0.60]
% assi
\draw [->] (0,-0.3) -- (0,4.5);
\draw [->] (-0.3,0) -- (5,0);
%rette
\draw (-0.3,3)--(4.5,3);
\draw (4,-0.3) -- (4,4.1);
%punti d'intersezione e poliedro
\draw [thick, fill, fill opacity = 0.1] plot [mark=*] coordinates {(0,0)(0,3)(4,3)(4,0)(0,0)};
%nodi
\node at (2,1.5) {$\mathcal{P}_s^{(\ell)}$};
%\node at (-0.3,-0.3) {$O$};
\node at (5.3,-0.45) {$P_{j}(f)$};
\node at (-0.3,4.8) {$P_{k}(f)$};
\node at (5,3)  {$P_D$};
\node at (4,4.5){$P_U$};
\end{tikzpicture}
\caption{$P_{k'}^{(\ell-1)}(f) = 0$}
\end{subfigure}
\caption{Geometric representation of the feasible set with different parameters. }
\label{fig:feas2}
\end{figure}
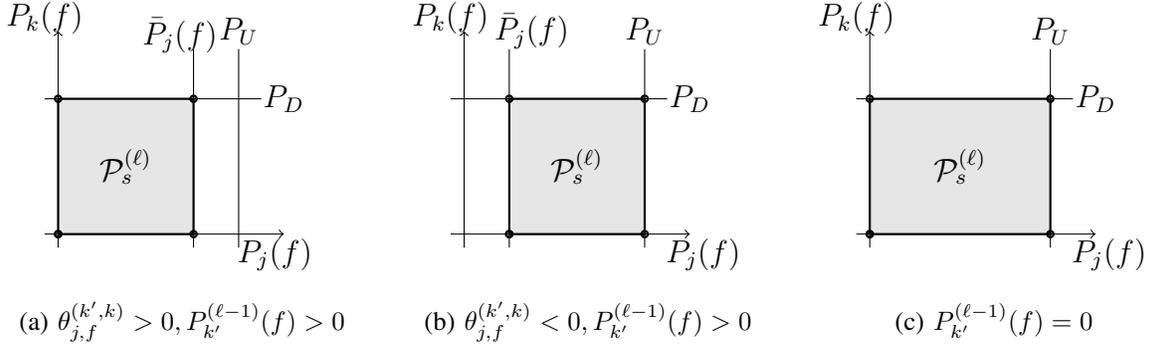

$L_{j,k,m}(\mb{P}(f), \bm{\mu})$ is a difference of convex function in $\mb{P}(f)$, i.e., it can be written as
\begin{equation}
L_{j,k,m}(\mb{P}(f), \bm{\mu}) = L_\text{cav}(\mathbf{P}(f)) + L_\text{vex}(\mathbf{P}(f),\bm{\mu}),
\end{equation}
where the expressions of the concave and convex parts are reported in~\eqref{eq:dc}.
Hence, problem~\eqref{eq:maxL} can be efficiently solved using the \emph{concave-convex procedure}~\cite{Yuille}. To elaborate, for each iteration $t$, the convex part of the function is linearized at the $t$-th stationary point $\mb{P}^{(t)}(f)$, obtaining a concave objective function. Then, the maximum can be easily found by a numerical solver.
At each iteration the problem becomes
\begin{equation}
\begin{aligned}
\label{eq:argmax}
\begin{bmatrix}
P_j^{(t+1)}(f) \\
P_k^{(t+1)}(f)
\end{bmatrix}
=& \argmax_{P_{j}(f),P_{k}(f) \in \mathcal{P}_{m}^{(\ell)}} \Big\{ L_\text{cav}(\mathbf{P}(f))%\\& \qquad 
+
\begin{bmatrix}
P_j(f) \\
P_k(f)
\end{bmatrix}^T \nabla_{(j,k)} L_\text{vex}(\mathbf{P}^{(t)}(f)) \Big\}
\end{aligned}
\end{equation}
where $\nabla_{(j,k)}$ represents the gradient respect to $P_{j}(f)$ and $P_{k}(f)$, reported in~\eqref{gradient}, while
\begin{equation}
\label{eq:intSetC}
\mc{C}(i,f) =
\begin{cases}
\mc{A}(f) \setminus i &x_{i,s}(f) = 1, \quad \forall i, \\
\mc{D}(f) &x_{i,w}(f) = 1, \quad i \in \mc{U},\\
\mc{U}(f) &x_{i,w}(f) = 1, \quad i \in \mc{D}.
\end{cases}
\end{equation} 
represents the set of users that receive interference from user $i$ on subcarrier $f$. 
\begin{figure*}[tb]
\smaller
\begin{equation}
\label{eq:dc}
\begin{aligned}
L_\text{cav}(\mathbf{P}(f)) &= \sum_{i \in \mc{A}(f)} \alpha_i \log_{2}\left( |h_{i,i}(f)|^2 P_{t}(f) + \sum_{n \in \mathcal{I}_i(f)} |h_{n,i}(f)|^2 P_{n}(f) + \sigma^2 \right) - K^{(\ell)} || \mb{P}(f) - \mb{P}^{(\ell-1)}(f) ||^2, \\
L_\text{vex}(\mathbf{P}(f)) &= - \sum_{i \in \mc{A}(f)} \alpha_i \log_{2} \left(\sum_{n \in \mathcal{I}_i(f)} |h_{n,i}(f)|^2 P_{i}(f) + \sigma^2 \right) - \sum_{j \in \mc{U}(f)} \mu_j P_{j}(f) -  \sum_{k \in \mc{D}(f)} \mu_0 P_{k}(f).
\end{aligned}
\end{equation}
\begin{equation}\label{gradient}
\begin{aligned}
\nabla_{(j,k)} L_\text{vex}(\mathbf{P}(f)) = -\frac{1}{\ln2} 
\begin{bmatrix}
\displaystyle\sum_{i \in \mc{C}_j(f)} \frac{\alpha_i |h_{j,i}(f)|^2}{\sum_{n \in \mc{I}_i(f)} |h_{n,i}(f)|^2 P_{n}(f) + \sigma^2} + \mu_j \ln2 \\
\displaystyle\sum_{i \in \mc{C}_k(f)} \frac{\alpha_i |h_{k,t}(f)|^2}{\sum_{n \in \mathcal{I}_i(f)} |h_{n,i}(f)|^2 P_{n}(f) + \sigma^2} + \mu_0 \ln2
\end{bmatrix}
\end{aligned}
\end{equation}
\begin{center}
\line(1,0){450}
\end{center}
\end{figure*}
The convergence of procedure~\eqref{eq:argmax} is given by the following theorem.
\begin{theorem}
\label{theo:cccp}
The iterative procedure given in~\eqref{eq:argmax} generates a sequence $\mb{P}^{(t)}(f)$ rendering $L_{j,k,m}(\mb{P}^{(t+1)}(f)) \geq L_{j,k,m}(\mb{P}^{(t)}(f))$, $\forall t > 0$, which finally converge on local optimum $\mb{P}^{*}(f) = \mb{P}^{(t+1)}(f) = \mb{P}^{(t)}(f)$.
\begin{proof}
The convergence can be demonstrated by the following inequality relations.
\begin{equation}
%\begin{aligned}
%L_{j,k,m}&(\mb{P}^{(t+1)}(f)) = L_\text{cav}(\mb{P}^{(t+1)}(f)) + L_\text{vex}(\mb{P}^{(t+1)}(f)) \\
%&> L_\text{cav}(\mb{P}^{(t+1)}(f))+ L_\text{vex}(\mb{P}^{(t)}(f)) \\ & \quad+ 
%\begin{bmatrix}
%P_{j}^{(t+1)}(f) - P_{j}^{(t)}(f)\\
%P_{k}^{(t+1)}(f) - P_{k}^{(t)}(f)
%\end{bmatrix}^T
%\nabla_{(j,k)} L_\text{vex}(\mb{P}^{(t)}(f)) \\
%&\geq L_\text{cav}(\mb{P}^{(t)}(f)) + 
%\begin{bmatrix}
%P_{j}^{(t)}(f) \\
%P_{k}^{(t)}(f)
%\end{bmatrix}^T
%\nabla_{(j,k)} L_\text{vex}(\mb{P}^{(t)}(f)) \\
%&\quad -
%\begin{bmatrix}
%P_{j}^{(t)}(f) \\
%P_{k}^{(t)}(f)
%\end{bmatrix}^T
%\nabla_{(j,k)} L_\text{vex}(\mb{P}^{(t)}(f)) + L_\text{vex}(\mb{P}^{(t)}(f)) \\
%&=L_\text{cav}(\mb{P}^{(t)}(f)) + L_\text{vex}(\mb{P}^{(t)}(f))  = L_{j,k,m}(\mb{P}^{(t)}(f))
%\end{aligned}
\begin{aligned}
L_{j,k,m}&(\mb{P}^{(t+1)}(f)) = L_\text{cav}(\mb{P}^{(t+1)}(f)) + L_\text{vex}(\mb{P}^{(t+1)}(f)) \\
&> L_\text{cav}(\mb{P}^{(t+1)}(f))+ L_\text{vex}(\mb{P}^{(t)}(f)) + 
\begin{bmatrix}
P_{j}^{(t+1)}(f) - P_{j}^{(t)}(f)\\
P_{k}^{(t+1)}(f) - P_{k}^{(t)}(f)
\end{bmatrix}^T
\nabla_{(j,k)} L_\text{vex}(\mb{P}^{(t)}(f)) \\
&\geq L_\text{cav}(\mb{P}^{(t)}(f)) + 
\begin{bmatrix}
P_{j}^{(t)}(f) \\
P_{k}^{(t)}(f)
\end{bmatrix}^T
\nabla_{(j,k)} L_\text{vex}(\mb{P}^{(t)}(f))  -
\begin{bmatrix}
P_{j}^{(t)}(f) \\
P_{k}^{(t)}(f)
\end{bmatrix}^T
\nabla_{(j,k)} L_\text{vex}(\mb{P}^{(t)}(f))\\  &\quad+ L_\text{vex}(\mb{P}^{(t)}(f)) =L_\text{cav}(\mb{P}^{(t)}(f)) + L_\text{vex}(\mb{P}^{(t)}(f))  = L_{j,k,m}(\mb{P}^{(t)}(f)),
\end{aligned} 
\nonumber
\end{equation}
where the first strict inequality came from the convexity of $L_\text{vex}$, and the second inequality came from~\eqref{eq:argmax}. The algorithm converge to local optimum $\mb{P}^*(f) = \mb{P}^{(t+1)}(f) = \mb{P}^{(t)}(f)$, because $\mb{P}^{(t)}(f) \in \mc{P}_{m}^{(\ell)}$, $\forall t > 0$, and the objective function is concave on a convex and bounded set.
\end{proof}
\end{theorem}

Finally, if $\theta_{j,f}^{(k',k)} = 0 $ or $\theta_{j,f}^{(k',k)} > 0$ and $\bar{P}_j(f) < 0$ or $\theta_{j,f}^{(k',k)} < 0$ and $\bar{P}_j(f) > P_U$, the feasible set~\eqref{eq:feasible} is empty.
Physically, this means that the strong downlink user cannot cancel the weak downlink user signal for any transmit power. Hence, it is possible to skip to an OMA configuration for the downlink, by imposing $P_{k}(f)= 0$, and computing the optimum uplink weak power $P_{j}(f)$ accordingly. 
This simplified mono-dimensional problem has the same expression of $L_\text{cav}$ and $L_\text{vex}$ given in \eqref{eq:dc}, where $P_{k}(f) = 0$. The procedure will converge to a local optimum $P_{j}(f)^*$, and its proof follows directly from Theorem~\ref{theo:cccp}.
%In this case, we can substitute problem~\eqref{eq:argmax} with the following reduced form
%\begin{equation} \label{eq:argmaxone}
%P_{j}^{(t+1)}(f) = \argmax_{0 \le P_{j}(f) \le P_U}  L_\text{cav}(P_{j}(f)) + P_{j}(f) \dpar{L_\text{vex}(P_{j}^{(t)}(f))}{P_{j}(f)}
%\end{equation}
%where the derivative of the convex part, at the $t$-th stationary point, is the first element of~\eqref{gradient}. The procedure~\eqref{eq:argmaxone} will converge to a local optimum $P_{j}(f)^*$, and its proof follows directly from Theorem~\ref{theo:cccp}.

\section{Dual variables updating}
\label{sec:dual}
\renewcommand{\theequation}{\thesection.\arabic{equation}}
\setcounter{equation}{0}
%Due to the non convexity of the optimization problems, an incorrect initial point $\mb{P}^{(0)}$ may lead the algorithm converging to a poor local optimum.
%A successful strategy to initiate the proposed procedures is to compute the power coefficients maximizing the auxiliary Lagrangian function $L$ without considering the interference of the other users, i.e., by setting the interfering link gains to 0. With this strategy, the allocation problem~\eqref{eq:maxLstrong} result convex and, thus, easy to solve. 
%Accordingly, initial dual variables $\bm{\mu}^{(0)}$ are chosen so that the power constraints are met.
%The rationale behind this strategy is that the iterations start  with the best \emph{single-user} allocation~\cite{Yu26}.

%Dual variables are updated with the same strategy for the three algorithms proposed: SA, WA, and PRA.
%To describe the updating process, let us fix the iteration $t$ and assuming computed the power coefficients $\mb{P}^{(t)*}$ and the allocation variables $\mb{x}^{(t)*}$, for a given value of $\bm{\mu}^{(t-1)}$.
%We are now interested in updating the Lagrange variable $\bm{\mu}^{(t)}$ of the dual problem.
%With this aim, the \emph{cutting plane ellipsoid} method is employed. For each iteration $t$, the algorithm computes the subgradient of $g(\bm{\mu})$ and updates the Lagrangian dual variables $\bm{\mu}$ according to the methodology provided in~\cite{Yu2006}.
%Even if the method is the same for SA, WA and PRA, the computation of the subgradients is different.
Dual variables are updated with the \emph{cutting plane ellipsoid} method for all the three algorithms SA, WA, and PRA. 
Let us focus on the update of the Lagrange variable $\bm{\mu}^{(t)}$, assuming computed the power coefficients $\mb{P}^{(t)*}$ and the allocation variables $\mb{x}^{(t)*}$, for a given value of $\bm{\mu}^{(t-1)}$. Hence, the algorithm computes the subgradient of $g(\bm{\mu})$ and updates the Lagrangian dual variables $\bm{\mu}$ according to the methodology provided in~\cite{Yu2006}.
In particular, the subgradient $d_{i,m}$, $i= 0,1,\dots, U$, $m \in \{s,w\}$ associated to the $i$-th element of $\bm{\mu}$ of the SA or WA algorithms is computed as
\begin{equation*}
\begin{aligned}
d_{0,m} =  P_{D,m}^{(\ell)} - \sum_{f=1}^{F} \sum_{k \in \mathcal{D}} x_{k,m}^*(f)  P_{k}^*(f), \quad%\\
d_{j,m} = P_{U,j,m}^{(\ell)} - \sum_{f=1}^{F} x_{j,m}^*(f) P_{j}^*(f), \quad \forall j \in \mathcal{U},
\end{aligned}
\end{equation*}
%On the other hand, the subgradient $d_{i}$, $i= 0,1,\dots, U$, associated to the $i$-th element of $\bm{\mu}$ of 
while for the PRA algorithm, it is computed as
\begin{equation*} 
%\label{eq:dualUpdate}
\begin{aligned}
d_{0}\hspace{-3pt} = \hspace{-3pt}P_D \hspace{-3pt} - \hspace{-3pt}\sum_{f=1}^{F} \sum_{k \in \mathcal{D}} \left( x_{k,s}^*(f) + x_{k,w}^*(f) \right)  P_{k}^*(f), \quad
d_{j}\hspace{-3pt} =\hspace{-3pt} P_U \hspace{-3pt}- \hspace{-3pt}\sum_{f=1}^{F} \left( x_{j,s}^*(f) + x_{j,w}^*(f) \right) P_{j}^*(f), \,\, \forall j \in \mathcal{U}.
\end{aligned}
\end{equation*}

\section{Proof of Theorem~\ref{theo:cccpSEQ}}
\renewcommand{\theequation}{\thesection.\arabic{equation}}
\setcounter{equation}{0}
\begin{proof}
\label{sec:dim2}
For each subcarrier, the optimization respect to $\mb{P}_D(f)$ will result in
\begin{equation} \label{eq:ine1}
%\begin{aligned}
%&L(\mb{P}_U^{(t)}(f), \mb{P}_D^{(t+1)}(f)) = \\ 
%&= L_\text{cav}(\mb{P_U}^{(t)}(f), \mb{P}_D^{(t+1)}(f))+ L_\text{vex} (\mb{P}_U^{(t)}(f), \mb{P}_D^{(t+1)}(f)) \\
%&> L_\text{cav}(\mb{P_U}^{(t)}(f), \mb{P}_D^{(t+1)}(f)) + L_\text{vex} (\mb{P}_U^{(t)}(f), \mb{P}_D^{(t)}(f)) \\ 
%&\,\,\, + (\mb{P}_D^{(t+1)}(f) - \mb{P}_D^{(t)}(f))^T \nabla_{\mb{P}_D(f)} L_\text{vex} (\mb{P}_U^{(t)}(f), \mb{P}_D^{(t)}(f)) \\ 
%&\geq L_\text{cav}(\mb{P}_U^{(t)}(f), \mb{P}_D^{(t)}(f)) + L_\text{vex} (\mb{P}_U^{(t)}(f), \mb{P}_D^{(t)}(f)) \\
%&\,\,\,+ \mb{P}_D^{(t)}(f)^\text{T} \nabla_{\mb{P}_D(f)} L_\text{vex} (\mb{P}_U^{(t)}(f), \mb{P}_D^{(t)}(f)) \\ 
%&\,\,\,  - \mb{P}_D^{(t)}(f)^\text{T} \nabla_{\mb{P}_D(f)} L_\text{vex} (\mb{P}_U^{(t)}(f), \mb{P}_U^{(t)}(f)) \\
%&= L_\text{cav}(\mb{P}_U^{(t)}(f), \mb{P}_D^{(t)}(f)) + L_\text{vex} (\mb{P}_U^{(t)}(f), \mb{P}_D^{(t)}(f)) \\
%&= L(\mb{P}_U^{(t)}(f), \mb{P}_D^{(t)}(f)),
%\end{aligned}
\begin{aligned}
&L(\mb{P}_U^{(t)}(f), \mb{P}_D^{(t+1)}(f))
= L_\text{cav}(\mb{P_U}^{(t)}(f), \mb{P}_D^{(t+1)}(f))+ L_\text{vex} (\mb{P}_U^{(t)}(f), \mb{P}_D^{(t+1)}(f)) \\
&> L_\text{cav}(\mb{P_U}^{(t)}(f), \mb{P}_D^{(t+1)}(f)) + L_\text{vex} (\mb{P}_U^{(t)}(f), \mb{P}_D^{(t)}(f)) \\ 
&\,\,\, + (\mb{P}_D^{(t+1)}(f) - \mb{P}_D^{(t)}(f))^T \nabla_{\mb{P}_D(f)} L_\text{vex} (\mb{P}_U^{(t)}(f), \mb{P}_D^{(t)}(f)) \\ 
&\geq L_\text{cav}(\mb{P}_U^{(t)}(f), \mb{P}_D^{(t)}(f)) + L_\text{vex} (\mb{P}_U^{(t)}(f), \mb{P}_D^{(t)}(f)) + \mb{P}_D^{(t)}(f)^\text{T} \nabla_{\mb{P}_D(f)} L_\text{vex} (\mb{P}_U^{(t)}(f), \mb{P}_D^{(t)}(f)) \\
&\,\,\,  - \mb{P}_D^{(t)}(f)^\text{T} \nabla_{\mb{P}_D(f)} L_\text{vex} (\mb{P}_U^{(t)}(f), \mb{P}_U^{(t)}(f)) \\ & = L_\text{cav}(\mb{P}_U^{(t)}(f), \mb{P}_D^{(t)}(f)) + L_\text{vex} (\mb{P}_U^{(t)}(f), \mb{P}_D^{(t)}(f)) = L(\mb{P}_U^{(t)}(f), \mb{P}_D^{(t)}(f)),
\end{aligned}
\end{equation}
where the first inequality derive from the strict convexity of $L_\text{vex}$ and the second derive from the first equation of sequential programming~\eqref{eq:seq}. 
Following the same approach, for the optimization of $\mb{P}_U(f)$ we obtain
\begin{equation} \label{eq:ine2}
L(\mb{P}_U^{(t+1)}(f), \mb{P}_D^{(t+1)}(f)) > L(\mb{P}_U^{(t)}(f), \mb{P}_D^{(t+1)}(f))
\end{equation}
Together, relations~\eqref{eq:ine1} and~\eqref{eq:ine2} will yield to $L(\mb{P}_U^{(t+1)}(f), \mb{P}_D^{(t+1)}(f)) \geq L(\mb{P}_U^{(t)}(f), \mb{P}_D^{(t)}(f))$.
Taking into account that the objective function is concave on a convex and bounded set and $\mb{P}_U^{(t)} \in \mc{P}_U$, $\mb{P}_D^{(t)} \in \mc{P}_D$, $\forall l > 0$, the algorithm converges to local optimum $\mb{P}(f) = [\mb{P}_U^*(f),\mb{P}_D^*(f)]$.
\end{proof}

\end{appendices}

\bibliographystyle{IEEEtran}
\bibliography{IEEEfull,FDNOMA}

\end{document}